\documentclass[natbib]{svjour3}
\usepackage{lmodern}
\usepackage{amssymb,amsmath}
\usepackage{ifxetex,ifluatex}
\usepackage{fixltx2e} 
\ifnum 0\ifxetex 1\fi\ifluatex 1\fi=0 
  \usepackage[T1]{fontenc}
  \usepackage[utf8]{inputenc}
\else 
  \ifxetex
    \usepackage{mathspec}
  \else
    \usepackage{fontspec}
  \fi
  \defaultfontfeatures{Ligatures=TeX,Scale=MatchLowercase}
  
\fi
\IfFileExists{upquote.sty}{\usepackage{upquote}}{}
\IfFileExists{microtype.sty}{%
\usepackage{microtype}
\UseMicrotypeSet[protrusion]{basicmath} 
}{}
\usepackage{hyperref}
\PassOptionsToPackage{usenames,dvipsnames}{color} 
\hypersetup{unicode=true,
            pdftitle={Machine Learning Guidance and Proof Certification for Connection Tableaux},
            pdfauthor={Michael~Färber; Cezary~Kaliszyk; Josef~Urban},
            pdfborder={0 0 0},
            breaklinks=true}
\usepackage{natbib}
\usepackage{listings}
\providecommand{\tightlist}{%
  \setlength{\itemsep}{0pt}\setlength{\parskip}{0pt}}
\title{Machine Learning Guidance and Proof Certification for Connection
Tableaux}
\author{Michael~Färber \and Cezary~Kaliszyk \and Josef~Urban}

\institute{M. Färber \at University of Innsbruck, Austria
\newline \email{michael.faerber@gedenkt.at} \and C. Kaliszyk \at University of Innsbruck, Austria
\newline \email{cezary.kaliszyk@uibk.ac.at} \and J. Urban \at Czech Technical University in Prague, Czech Republic
\newline \email{josef.urban@gmail.com}}

\journalname{J. Automated Reasoning}

\date{}
\usepackage{centerfloat}
\usepackage{tikz}
\usetikzlibrary{calc,tikzmark}
\usepackage{tikz-qtree}
\usepackage{subfig}
\usepackage{todonotes}

\usepackage{pgfplots}
\pgfplotsset{compat=1.9}
\usepgfplotslibrary{units}

\usepackage{booktabs}
\usepackage{prftree}

\usepackage{makecell}

\usepackage[strings]{underscore}

\usepackage{floatlistings}

\usepackage{operators}

\ifx\paragraph\undefined\else
\let\oldparagraph\paragraph
\renewcommand{\paragraph}[1]{\oldparagraph{#1}\mbox{}}
\fi
\ifx\subparagraph\undefined\else
\let\oldsubparagraph\subparagraph
\renewcommand{\subparagraph}[1]{\oldsubparagraph{#1}\mbox{}}
\fi

\begin{document}
\maketitle
\begin{abstract}
Connection calculi allow for very compact implementations of
goal-directed proof search. We give an overview of our work related to
connection tableaux calculi: First, we show optimised functional
implementations of clausal and nonclausal proof search, including a
consistent Skolemisation procedure for machine learning. Then, we show
two guidance methods based on machine learning, namely reordering of
proof steps with Naive Bayesian probablities, and expansion of a proof
search tree with Monte Carlo Tree Search. Finally, we give a translation
of connection proofs to LK, enabling proof certification and automatic
proof search in interactive theorem provers.
\end{abstract}

\section{Introduction}\label{introduction}

Connection calculi have enabled proof search in a variety of logics:
First-order automated theorem provers (ATPs) based on connection calculi
have been implemented for classical (leanCoP \citep{otten2008-leancop}),
intuitionistic (ileanCoP \citep{otten2005-ileancop}), and modal logic
(MleanCoP \citep{otten2014-mleancop}). A variant of leanCoP with
interpreted linear arithmetic (leanCoP-\(\Omega\)) won the TFA division
of CASC-J5 \citep{sutcliffe2011-j5}. Furthermore, nanoCoP
\citep{otten2016-nanocop} is a version of leanCoP able to perform proof
search without clausification.

The size of these provers is only in the range of hundreds of
bytes.\footnote{The prover at
  \url{http://www.leancop.de/programs/veryleancop.pl} is as small as 199
  bytes.} This makes these provers suitable as a basis for experiments
and adaptions, such as \emph{machine learning} and \emph{proof
certification}. For these applications, we implemented connection
provers in functional instead of logic programming languages. There are
several reasons: First, a large number of interactive theorem provers
(ITPs), such as HOL Light \citep{harrison2009-hollight}, HOL4
\citep{slind2008-hol4}, Isabelle \citep{wenzel2008-isabelle}, Coq
\citep{bertot2008-coq}, and Agda \citep{bove2009-agda} are written in
functional programming languages, lending themselves well to integration
of functional proof search tactics. Second, several machine learning
algorithms have been recently implemented efficiently for these ITPs in
the functional languages. Third, we achieve better performance with
functional-style implementations, which is important to compensate for
the performance penalty incurred by machine learning.

The machine learning connection provers MaLeCoP
\citep{urban2011-malecop} and FEMaLeCoP \citep{kaliszyk2015-femalecop}
integrate Naive Bayesian clause ordering into leanCoP's proof search.
The monteCoP \citep{faerber2017-montecop} system implements Monte Carlo
Tree Search in leanCoP. The resulting connection proofs can be certified
in interactive theorem provers \citep{kaliszyk2015-holcop}, giving at
the same time rise to proof search tactics similar to Metis
\citep{hurd2003-metis} or MESON \citep{harrison1996-meson}.

In this paper we develop an integration of internal guidance based on
machine learning and Monte Carlo methods in connection-style proof
search, as well as how to certify the resulting connection calculus
proofs. The contributions described in this paper are:

\begin{itemize}
\tightlist
\item
  Efficient implementations of proof search for the clausal and
  nonclausal connection calculi in functional programming languages, see
  \autoref{funconn}.
\item
  Integration of machine learning based internal guidance, see
  \autoref{bayesguidance}. We integrate a multi-layer indexing based on
  previous proofs suitable for a Naive Bayes classifier of extension
  steps. These rely on consistent symbol names
  (\autoref{skolemisation}). We further provide various strategies to
  improve the proof search based on the learned information.
\item
  Integration of Monte Carlo Tree Search, see
  \autoref{monte-carlo-proof-search}. This includes a number of proof
  state evaluation heuristics, some learned on previous proofs.
\item
  A unified formulation of clausal and nonclausal connection calculi
  adapted for proof translation, and proof certification integrated into
  HOL Light, see \autoref{certification}.
\end{itemize}

The paper combines and extends our works presented at CPP 2015
\citep{kaliszyk2015-holcop}, LPAR 2015 \citep{kaliszyk2015-femalecop},
and CADE 2017 \citep{faerber2017-montecop}. The techniques added over
the conference versions include: the higher-order logic reconstruction
of nonclausal proofs; consistent Skolemisation applicable also for
nonclausal proof search; and efficient functional-style implementation
of proof search in clausal and nonclausal connection calculi.\footnote{The
  source code of all implementations in this article is available at
  \url{http://cl-informatik.uibk.ac.at/users/mfaerber/cop.html}.}

\section{Connection Tableaux Calculi}\label{conncalculus}

Connection calculi provide a goal-oriented way to search for proofs in
classical and nonclassical logics \citep{otten2008-leancop}. Common to
these calculi is the concept of connections \(\{P, \lnot P\}\) between
literals \(P\) and \(\lnot P\), which correspond to closing a branch in
the tableaux calculus \citep{haehnle2001-tableaux}.

Connection tableaux calculi, such as \citep{bibel1983-matings}, are
members of the family of connection calculi. As the calculi considered
in this paper have a very small set of rules, they lend themselves very
well to proof translation and machine learning.

In this section, we introduce the \emph{clausal} and the
\emph{nonclausal connection calculus} that we will use throughout the
paper.

\subsection{Connection Calculi}\label{connection-calculi}

The connection calculi in this paper operate on \emph{matrices}, where a
matrix is a set of clauses. In the nonclausal calculus, clauses do not
only contain literals, but also matrices, giving rise to a nested
structure. \(M\) stands for a matrix, \(C\) for a clause, \(L\) for a
literal, \(x\) for a variable, and \(\vec x\) for a sequence of
variables, as in \(\forall \vec x. P(\vec x)\). A substitution
\(\sigma\) is a mapping from variables to terms. The \emph{complement}
\(\overline L\) is \(A\) if \(L\) has the shape \(\lnot A\), otherwise
\(\overline L\) is \(\lnot A\). A \emph{\(\sigma\)-complementary
connection} \(\{L, L'\}\) exists if \(\sigma \overline L = \sigma L'\).
Given a relation \(R\), its transitive closure is denoted by \(R^+\) and
its transitive reflexive closure by \(R^*\).

We will focus on two variants of the calculus: clausal and nonclausal.
As the two alternatives will differ only in the clauses and rules, we
first give a definition of the common parts of the clausal and
nonclausal connection calculi, omitting the calculus rules.

\begin{definition}[Connection Calculus]A connection calculus is a
calculus satisfying the following conditions. The words of the calculus
are tuples \(\langle C, M, Path \rangle\), where \(C\) is a clause,
\(M\) is a matrix, and \(Path\) is a set of literals called the
\emph{active path}. \(C\) and \(Path\) can also be empty, denoted
\(\varepsilon\). In the rules of the calculus, \(\sigma\) is a term
substitution, and \(\{L, L'\}\) is a \(\sigma\)-complementary
connection. The substitution \(\sigma\) is global (or \emph{rigid}),
i.e.~it is applied to the whole derivation.\end{definition}

To complete the definitions of the variants of the connection calculus,
we need to specify the types of clauses and the rules. In the clausal
connection calculus, a clause is a set of literals. The calculus rules
are presented in \autoref{fig:clausal-calculus}.

\begin{figure}
  \begin{tabular}{r m{4cm}}
  \parbox[c]{\hsize}{
    \begin{prfenv}
      \prftree[r]{Axiom}
      {}
      {\{\}, M, Path}
    \end{prfenv}}
  &
  \\[1em]
  \parbox[c]{\hsize}{
    \begin{prfenv}
      \prftree[r]{Start}
      {\prfassumption{C_2, M, \{\}}}
      {\varepsilon, M, \varepsilon}
    \end{prfenv}}
  & where $C_2$ is copy of $C_1 \in M$
  \\
  \parbox[c]{\hsize}{
    \begin{prfenv}
      \prftree[r]{Reduction}
      {C           , M, Path \cup \{L'\}}
      {C \cup \{L\}, M, Path \cup \{L'\}}
    \end{prfenv}}
  & where $\sigma(L) = \sigma(\overline{L'})$
  \\
  \parbox[c]{\hsize}{
    \begin{prfenv}
      \prftree[r]{Extension}
      {\prfassumption{C_2 \setminus \{L'\}, M, Path \cup \{L\}}}
      {\prfassumption{C                   , M, Path}}
      {C \cup \{L\}, M, Path}
    \end{prfenv}}
  & where
  $C_2$ is copy of $C_1 \in M$ and
  $L' \in C_2$ with $\sigma(L) = \sigma(\overline{L'})$
\end{tabular}

  \caption{Clausal connection calculus rules.}
  \label{fig:clausal-calculus}
\end{figure}

In the nonclausal connection calculus, a clause is a set of literals and
matrices. The following definitions of the concepts used in the
extension rule follow Otten
\citep{otten2011-nonclausal, otten2016-nanocop}.

\begin{definition}[Clause Predicates]\label{def:clausepreds}A clause
\(C\) \emph{contains} \(L\) iff \(L \in ^+ C\). A clause \(C \in ^+ M\)
is \emph{\(\alpha\)-related} to a literal \(L\) iff \(M' \in ^* M\) with
\(\{C_L, C_C\} \subseteq M'\) such that \(C_L \neq C_C\),
\(L \in ^+ C_L\), and \(C \in ^* C_C\). \(C'\) is a \emph{parent clause}
of \(C\) iff \(M' \in C'\) and \(C \in M'\) for some matrix
\(M'\).\end{definition}

\begin{definition}[Clause Functions]\label{def:clausefuns}A \emph{copy
of the clause} \(C\) in the matrix \(M\) is created by replacing all
\emph{free variables} in \(C\) with fresh variables.
\(M[C_1 \backslash C_2]\) denotes the matrix \(M\) in which the clause
\(C_1\) is replaced by the clause \(C_2\).\end{definition}

\begin{definition}[Extension Clause,
\(\beta\)-clause]\label{def:ebclause}\(C\) is an \emph{extension clause}
(\emph{e-clause}) of the matrix \(M\) with respect to a set of literals
\(Path\) iff either (a) \(C\) contains a literal of \(Path\), or (b)
\(C\) is \(\alpha\)-related to all literals of \(Path\) occurring in
\(M\) and if \(C\) has a parent clause, that parent clause contains a
literal of \(Path\). The \emph{\(\beta\)-clause} of \(C_2\) with respect
to \(L_2\) is \(C_2\) with \(L_2\) and all clauses that are
\(\alpha\)-related to \(L_2\) removed.\end{definition}

The rules are presented in \autoref{fig:nonclausal-calculus}. The
difference in the calculus rules to the clausal variant is the addition
of a decomposition rule, and the adaptation of the extension rule to the
nonclausal setting.

\begin{figure}
  \begin{tabular}{r m{4.5cm}}
  \parbox[c]{\hsize}{
    \begin{prfenv}
      \prftree[r]{Axiom}
      {}
      {\{\}, M, Path}
    \end{prfenv}}
  &
  \\[1em]
  \parbox[c]{\hsize}{
    \begin{prfenv}
      \prftree[r]{Start}
      {\prfassumption{C_2, M, \{\}}}
      {\varepsilon, M, \varepsilon}
    \end{prfenv}}
  & where $C_2$ is copy of $C_1 \in M$
  \\
  \parbox[c]{\hsize}{
    \begin{prfenv}
      \prftree[r]{Reduction}
      {C           , M, Path \cup \{L'\}}
      {C \cup \{L\}, M, Path \cup \{L'\}}
    \end{prfenv}}
  & where $\sigma(L) = \sigma(\overline{L'})$
  \\
  \parbox[c]{\hsize}{
    \begin{prfenv}
      \prftree[r]{Extension}
      {\prfassumption{C_3, M[C_1 \backslash C_2], Path \cup \{L\}}}
      {\prfassumption{C  , M                    , Path}}
      {C \cup {L}, M, Path}
    \end{prfenv}}
  & where
    $C_3$ is the $\beta$-clause of $C_2$ wrt. $L'$,
    $C_2$ is copy of $C_1$, $C_1$ is e-clause of $M$ wrt. $Path \cup \{L\}$,
    $C_2$ contains $L'$ with $\sigma(L) = \sigma(\overline{L'})$
  \\
  \parbox[c]{\hsize}{
    \begin{prfenv}
      \prftree[r]{Decomposition}
      {\prfassumption{C \cup C', M, Path}}
      {C \cup {M'}, M, Path}
    \end{prfenv}}
  & where $C' \in M'$
\end{tabular}

  \caption{Nonclausal connection calculus rules.}
  \label{fig:nonclausal-calculus}
\end{figure}

\begin{example}\label{ex:conn-calc}Consider the following formula \(F\)
and its conjunctive normal form \(F'\). We will attempt to show that
they imply \(\bot\): \[\begin{aligned}
  F  &= Q \land P(a) \land
\forall x. (\lnot P(x) \lor (\lnot P(s^2 x) \land (P(sx) \lor \lnot Q))), \\
  F' &= \forall x.
(Q \land P(a) \land
(\lnot P(x) \lor \lnot P(s^2 x)) \land
(\lnot P(x) \lor P(sx) \lor \lnot Q)).
  \end{aligned}\] For brevity, we write \(sx\) for \(s(x)\) and
\(s^2 x\) for \(s(s(x))\). The nonclausal matrix \(M\) corresponds to
\(F\) and the clausal matrix \(M'\) to \(F'\):

\begin{small}
  $M = \left[
  [Q]
  [P(a)]
  \left[\begin{matrix}
    [\lnot P(x)] \\
    \left[
      [\lnot P(s^2 x)]
      \left[\begin{matrix} P(sx) \\ \lnot Q \end{matrix}\right]
    \right] \\
  \end{matrix}\right]
\right]
, \quad M' = \left[
[Q]
[P(a)]
\left[\begin{matrix} \lnot P(x) \\ \lnot P(s^2 x) \end{matrix}\right]
\left[\begin{matrix} \lnot P(x) \\ P(sx) \\ \lnot Q \end{matrix}\right]
\right]
.$
  \end{small}

Graphical proofs for the problem are given in
\autoref{fig:matrix-nonclausal} respectively
\autoref{fig:matrix-clausal}: There, lines represent connections, and
the substitution used is
\(\sigma = \{x' \mapsto a, \hat x \mapsto sx', \bar x \mapsto x'\}\). A
formal proof for \(M'\) in the clausal connection calculus is given in
\autoref{fig:clausal-calc-ex}. A shorter proof for \(M'\) as well as a
formal proof for \(M\) will be given using slightly modified versions of
the calculi in \autoref{conn-calc-recon}.\end{example}

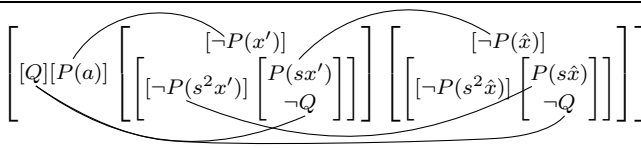
\begin{figure}

\newcommand{\ncmatrixmark}[1]{\tikzmark{ncm#1}}
\newcommand{\ncmatrixprefix}[1]{pic cs:ncm#1}

\newcommand{\ncmatrix}[1]{
\left[
  [Q #1{q}]
  [P(a) #1{pa}]
  \left[\begin{matrix}
    [\lnot P(x') #1{npx1}] \\
    \left[
      [\lnot P(s^2 x') #1{npssx1}]
      \left[\begin{matrix}
        P(sx') #1{psx1} \\
        \lnot Q #1{nq1}
      \end{matrix}\right]
    \right] \\
  \end{matrix}\right]
  \left[\begin{matrix}
    [\lnot P(\hat x) #1{npx2}] \\
    \left[
      [\lnot P(s^2 \hat x)]
      \left[\begin{matrix}
        P(s \hat x) #1{psx2} \\
        \lnot Q #1{nq2}
      \end{matrix}\right]
    \right] \\
  \end{matrix}\right]
\right]
}

\newcommand{\ncmatrixconn}[1]{
\begin{tikzpicture}[overlay,remember picture]
  \coordinate (Q)      at ($(#1{q})+(-0.3em,-0.3em)$);
  \coordinate (NQ1)    at ($(#1{nq1})+(-0.5em, -0.3em)$);
  \coordinate (NQ2)    at ($(#1{nq2})+(-0.5em, -0.3em)$);
  \coordinate (PA)     at ($(#1{pa})+(-1.3em,0.9em)$);
  \coordinate (NPX1)   at ($(#1{npx1})+(-2.5em, 0.5em)$);
  \coordinate (NPX2)   at ($(#1{npx2})+(-2.2em, 0.6em)$);
  \coordinate (NPSSX1) at ($(#1{npssx1})+(-2.5em,-0.3em)$);
  \coordinate (PSX1)   at ($(#1{psx1})+(-1.5em,0.9em)$);
  \coordinate (PSX2)   at ($(#1{psx2})+(-2.5em, -0.3em)$);
  \coordinate (QNQ21)   at ($(NPSSX1)+(-1em, -1.75em)$);
  \coordinate (QNQ22)   at ($(NQ2)+(-3em, -1em)$);
  \draw (Q)      to [bend angle=30, bend right] (NQ1);
  \draw (PA)     to [bend angle=50, bend left] (NPX1);
  \draw (NPSSX1) to [bend angle=25, bend right](PSX2);
  \draw (PSX1)   to [bend angle=40, bend left] (NPX2);
  \draw plot [smooth] coordinates {(Q) (QNQ21) (QNQ22) (NQ2)};
\end{tikzpicture}
}
  $\ncmatrix{\ncmatrixmark}$
  \ncmatrixconn{\ncmatrixprefix}
  \caption{Nonclausal graphical connection proof.}
  \label{fig:matrix-nonclausal}
\end{figure}

\begin{figure}

\newcommand{\cmatrixmark}[1]{\tikzmark{cm#1}}
\newcommand{\cmatrixprefix}[1]{pic cs:cm#1}

\newcommand{\cmatrix}[1]{
\left[
  [Q #1{q}]
  \left[\begin{matrix}
    \lnot P(x') #1{npx1} \\
    P(sx') #1{psx1} \\
    \lnot Q #1{nq1}
  \end{matrix}\right]
  [P(a) #1{pa}]
  \left[\begin{matrix}
    \lnot P(\hat x) #1{npx2} \\
    P(s \hat x) #1{psx2} \\
    \lnot Q #1{nq2}
  \end{matrix}\right]
  \left[\begin{matrix}
    \lnot P(\bar x) #1{npx3} \\
    \lnot P(s^2 \bar x) #1{npssx}
  \end{matrix}\right]
\right]
}

\newcommand{\cmatrixconn}[1]{
\begin{tikzpicture}[overlay,remember picture]
  \coordinate (Q)      at ($(#1{q})+(-0.3em,-0.3em)$);
  \coordinate (NQ1)    at ($(#1{nq1})+(-0.5em, -0.3em)$);
  \coordinate (NQ2)    at ($(#1{nq2})+(-0.5em, -0.3em)$);
  \coordinate (PA)     at ($(#1{pa})+(-1.3em,0.9em)$);
  \coordinate (NPX1)   at ($(#1{npx1})+(-2.5em, 0.6em)$);
  \coordinate (NPX2)   at ($(#1{npx2})+(-2.3em, 0.5em)$);
  \coordinate (NPX3)   at ($(#1{npx3})+(-2.3em, 0.5em)$);
  \coordinate (NPSSX)  at ($(#1{npssx})+(-2.5em,-0.3em)$);
  \coordinate (PSX1)   at ($(#1{psx1})+(0em,0.9em)$);
  \coordinate (PSX2)   at ($(#1{psx2})+(0em, -0.3em)$);
  \draw (Q)      to [bend angle=30, bend right] (NQ1);
  \draw (Q)      to [out=-50, in=-160] (NQ2);
  \draw (NPX1)   to [bend angle=50, bend left]   (PA);
  \draw (PSX2)   to [bend angle=60, bend right](NPSSX);
  \draw (PSX1)   to [bend angle=40, bend left] (NPX2);
  \draw (NPX3)   to [bend angle=40, bend right] (PA);
\end{tikzpicture}
}
  $\cmatrix{\cmatrixmark}$
  \cmatrixconn{\cmatrixprefix}
  \caption{Clausal graphical connection proof.}
  \label{fig:matrix-clausal}
\end{figure}

\begin{figure}
  $\input{prftree/clausal-calc-ex2}$
  \caption{Clausal connection proof.}
  \label{fig:clausal-calc-ex}
\end{figure}

Soundness and completeness have been proved both for the clausal
\citep{letz2001-connection} and for the nonclausal calculus
\citep{otten2011-nonclausal}. We will discuss practical functional-style
implementations of proof search for the presented calculi in
\autoref{funconn}.

\section{Functional-style Connection Prover}\label{funconn}

In this section, we develop an efficient implementation of a connection
prover for classical first-order logic in a functional programming
language. The resulting implementation will be the basis for all
experiments in the remainder of the paper.

The connection prover performs the following tasks. Given a classical
first-order logic problem, it creates a matrix for the problem. The
matrix is then used to build an index (usually in the form of a literal
database) which will be used to provide an efficient way to find
connections during proof search. Finally, proof search with iterative
deepening is performed. We detail these steps in the next sections.

\subsection{Problem Preprocessing}\label{preprocessing}

We focus on first-order logic problems represented as a set of axioms
\(\{A_1, \dots, A_n\}\) together with a conjecture \(C\), where all
axioms and the conjecture are closed formulas. The goal is to show that
the axioms imply the conjecture. For convenience, in the actual
implementation we use the TPTP format \citep{sutcliffe2009-tptp} as
input. Each parsed input problem is transformed according to the
following procedure. Only the steps 2 and 6 differ in comparison with
the original Prolog implementations of leanCoP and nanoCoP
\citep{otten2008-leancop, otten2016-nanocop}.

\begin{enumerate}
\def\labelenumi{\arabic{enumi}.}
\tightlist
\item
  The conjecture \(C\) is combined with the axioms
  \(\{A_1, \dots, A_n\}\) to form the new problem
  \((A_1 \land \dots \land A_n) \rightarrow C\) (respectively \(C\) if
  no axioms are present).
\item
  Constants and variables are mapped to integers, to enable more
  efficient lookup and comparison during the proof search, as needed
  e.g.~for fast unification.
\item
  As the connection tableaux calculi considered in this paper do not
  have special rules for equality, equality axioms are added to the
  problem if equality appears in the original problem. The axioms are:
  reflexivity, symmetry, transitivity, and:

  \begin{itemize}
  \tightlist
  \item
    For every \(n\)-ary function \(f\), the formula
    \(x_1 = y_1 \rightarrow \dots \rightarrow x_n = y_n \rightarrow  f(x_1, \dots, x_n) = f(y_1, \dots, y_n)\)
    is introduced.
  \item
    For every \(n\)-ary predicate \(P\), the formula
    \(x_1 = y_1 \rightarrow \dots \rightarrow x_n = y_n \rightarrow  P(x_1, \dots, x_n) \rightarrow P(y_1, \dots, y_n)\)
    is introduced.
  \end{itemize}
\item
  If the formula has the shape \(P \rightarrow C\), then it is
  transformed to the equivalent
  \((P \land \#) \rightarrow (C \land \#)\). \(\#\) is a marker that can
  be understood to be equivalent to \(\top\). It allows proof search to
  recognise clauses stemming from the conjecture (see
  \citep[sec.~2.1]{otten2008-leancop}).
\item
  Implications and equivalences are expanded, e.g. \(A \rightarrow B\)
  becomes \(\lnot A \lor B\).
\item
  Quantifiers are pushed inside such that their scope becomes minimal. 
\item
  The formula is negated (to perform a proof by refutation) and
  converted to negation normal form.
\item
  The formula is reordered to minimise the number of paths through the
  matrix.
\item
  The formula is rectified, i.e.~variables are renamed such that any two
  distinct quantifiers in the formula bind variables with different
  names.
\item
  The formula is skolemised. For machine learning, we use consistent
  Skolemisation as discussed in \autoref{skolemisation} instead of outer
  Skolemisation as performed in the original Prolog version.
\end{enumerate}

\subsection{Matrix and Literal Database}\label{matlitdb}

The matrix is built from the skolemised formula resulting from
\autoref{preprocessing}. For the clausal connection prover, this
involves a transformation of the formula into clausal normal form. The
\emph{standard transformation} applies distributivity rules of the shape
\(A \land (B \lor C) \equiv (A \lor B) \land (A \lor C)\) to the formula
until a fix point is reached. In the worst case, this transformation
makes the formula grow exponentially. To avoid this, the
\emph{definitional transformation} introduces new symbols
\citep{plaisted1986-cnf, otten2010-cut}. Similarly to Skolemisation, the
introduced symbols should be consistent across different problems, which
is achieved by using a normalised string representation of the clause
literals as new symbol names. For the nonclausal connection prover, no
clausification is required, as the formula can be directly transformed
into the nonclausal matrix. For both clausal and nonclausal matrices,
the polarity of literals is encoded by the sign of the integer
representing the predicate symbol.

We next explain how the prover efficiently searches for connections.
leanCoP and nanoCoP rely on Prolog's internal literal indexing. For
every literal \(L\) a set of \emph{contrapositives} is stored, where a
contrapositive corresponds to branches that have to be opened in order
to close the leaf \(\overline L\). Finding a connection then amounts to
finding contrapositives for \(L\) such that \(L\) can be unified with
the negation of the leaf of the current branch. In the clausal case, the
contrapositive for a literal \(L\) in the clause \(C\) is
\(C \setminus L\). In the nonclausal case, the contrapositive for \(L\)
is a copy of the matrix in which all clauses \(\alpha\)-related to \(L\)
are deleted, allowing for the efficient construction of e-clauses. In
our implementation, we store contrapositives in a hash table indexed by
the root symbols of literals. This allows efficient search for literals
that have the same root as the leaf of the current branch, but with
opposite polarity.

We also considered storing contrapositives in first-order term indexing
structures \citep{ramakrishnan2001-indexing}. The overall effect on
performance of storing contrapositives in a discrimination tree
\citep{greenbaum1986-phd} on the considered datasets is however minor,
as unification with array substitutions (see \autoref{unification}) is
relatively fast.

\subsection{Proof Search}\label{proof-search}

In both clausal and nonclausal calculi, proof search is analytic,
i.e.~the proof tree is constructed bottom-up. As the proof search is not
confluent, i.e.~making a wrong choice can lead to a dead end,
backtracking is necessary for completeness. The proof tree is
constructed with a depth-first strategy, which results in an incomplete
proof search. To remedy this, iterative deepening is used, where the
maximal path length \lstinline!lim! is increased in every iteration.

The principal implementations of the connection tableaux calculi,
leanCoP and nanoCoP, use a number of optimisation techniques, such as
regularity, lemmata, and restricted backtracking \citep{otten2010-cut}.
When backtracking is restricted, as soon as the proof search finds some
proof tree to close a branch, no other potential proof trees for that
branch are considered any more. While restricted backtracking loses
completeness, it significantly increases the number of problems solved
for various first-order problem classes.

Prolog allows for a very elegant and succinct implementation of proof
search. First attempts to directly integrate machine learning into
Prolog leanCoP have suffered from slow speed \citep{urban2011-malecop}.
We later showed \citep{kaliszyk2015-femalecop, kaliszyk2015-holcop} that
implementations of leanCoP in a functional programming language allow
for fast machine learning as well as for efficient proof reconstruction
in interactive theorem provers. However, implementing proof search with
restricted backtracking in a functional language is not straightforward.

In this section, we discuss several implementations of a clausal prover
loop that permits restricted backtracking: The simplified version of
leanCoP shown in \autoref{prolog} is the smallest, but also the slowest
implementation. Care is taken that all subsequent implementations
perform the proof search in precisely the same order as the original
Prolog implementation. We then introduce a purely functional
implementation in \autoref{lazylist} using lazy lists respectively
streams. This version slightly increases code size compared to the
Prolog version, but greatly improves performance, as shown in the
evaluation in \autoref{funconn-eval}. We also discuss an approach based
on continuations, still purely functional, but more complicated than the
stream version. In exchange, this version has slightly better
performance than the stream one, likely due to not having to allocate
memory for (stream) constructors. The fastest, but also most complicated
implementation considered in this paper uses an explicit stack and
exceptions for backtracking. However, as it proves only as many problems
as the continuation-based solution, we will only briefly discuss it.

\subsubsection{Prolog}\label{prolog}

A simplified version of the original leanCoP in Prolog is given in
\autoref{lst:leancop-pl}. We explain and relate it to the clausal
connection calculus introduced in \autoref{conncalculus}.

The main predicate \lstinline!prove(!\(C\)\lstinline!,!
\(Path\)\lstinline!, PathLim)! succeeds iff there exists a closed proof
tree for \(C, M, Path\) with a maximal \(Path\) length of
\lstinline!PathLim!. For this, \lstinline!prove! attempts to close the
proof tree for the first literal \lstinline!Lit! of \(C\) in lines 4--9,
and if successful, it continues with the remaining clause
\lstinline!Cla! of \(C\) in line 10.

Let us detail the proof search for the current literal \lstinline!Lit!:
Line 4 corresponds to the \emph{reduction} rule: The branch is closed if
the negation of \lstinline!Lit! can be unified with a literal on the
\(Path\). Lines 6--8 correspond to the \emph{extension} rule: The
literal database as explained in \autoref{matlitdb} is implemented by
the predicate \lstinline!lit(!\(L\)\lstinline!,! \(C\)\lstinline!)!,
which succeeds iff the matrix contains some clause that can be unified
with \(\{L\} \cup C\). This is used to obtain some contrapositive
\lstinline!Cla1! for the negation of \lstinline!Lit!. If the path does
not exceed the length limit (line 7), new branches are opened for
\lstinline!Cla1! in line 8.

Backtracking is handled by the Prolog semantics: For example, if
choosing the first matching contrapositive for \lstinline!Lit! leads to
the proof search getting stuck, the next contrapositive will be tried by
Prolog.

\begin{lstlisting}[language=Prolog, caption=Clausal proof search in Prolog., label=lst:leancop-pl]
prove([],_,_).
prove([Lit|Cla],Path,PathLim) :-
    (-NegLit=Lit;-Lit=NegLit) ->
       ( member(NegL,Path), unify_with_occurs_check(NegL,NegLit)
         ;
         lit(NegLit,Cla1),
         ( length(Path,K), K<PathLim -> true ; fail ),
         prove(Cla1,[Lit|Path],PathLim)
       ),
       prove(Cla,Path,PathLim).
\end{lstlisting}

\subsubsection{Lazy Lists and Streams}\label{lazylist}

Proof search in a functional language can be elegantly implemented as a
function from a branch to a \emph{lazy list} of proofs, where a lazy
list is an arbitrarily long list built on demand. However, as the proof
search considers every list element maximally once, the memoization done
for lazy lists creates an unnecessary overhead. For that reason,
\emph{streams} can be used instead of lazy lists, where a stream is a
special case of a lazy list that restricts list elements to be traversed
maximally once. As our application uses a common interface for lazy
lists and streams, we solely present the lazy list version here.

\autoref{lst:lazycop} shows a functional leanCoP implementation using
lazy lists. Let us first introduce the semantics of the used constructs:

\begin{itemize}
\tightlist
\item
  \lstinline!x & f! denotes \lstinline!f x!.
\item
  \lstinline!\ x -> y! stands for a lambda expressions \(\lambda x. y\).
\item
  \lstinline!unify sub lit1 lit2! unifies two literals \lstinline!lit1!
  and \lstinline!lit2! under a substitution \lstinline!sub!, returning a
  new substitution if successful.
\item
  \lstinline!unifyDB sub lit! finds all contrapositives in the literal
  database which could match the literal \lstinline!lit! under the
  substitution \lstinline!sub!. It returns a list of
  substitution-contrapositive pairs. It corresponds to the
  \lstinline!lit! predicate in the Prolog version.
\item
  \lstinline!mapMaybe f l! returns the results of \lstinline!f! for the
  elements of \lstinline!l! on which \lstinline!f! succeeded.
\item
  \lstinline!concatMap f l! maps \lstinline!f! over all elements of
  \lstinline!l! and concatenates the resulting list of lists to form a
  flat list.
\item
  \lstinline!x ++ y! is the concatenation of two lists \lstinline!x! and
  \lstinline!y!.
\end{itemize}

The main function \lstinline!prove! \(C\) \(Path\) \lstinline!lim!
\(\sigma\) returns a list of substitutions
\([\sigma _1, \dots, \sigma _n]\), where every substitution
\(\sigma _i\) corresponds to a closed proof tree for \(C, M, Path\) with
a maximal path length smaller than \lstinline!lim!, where the global
initial substitution is \(\sigma\) and the final substitution is
\(\sigma _i\).\footnote{In this simplified implementation, the actual
  proof tree is not recorded, in contrast to our actual implementation.
  The same holds for the Prolog version.} Similarly to the Prolog
version, \lstinline!prove! attempts to close the proof tree for the
first literal \lstinline!lit! of \(C\) in lines 4--8, and the resulting
substitutions are used to close the proof trees for the remaining clause
\lstinline!cla! of \(C\) in line 9. Line 4 corresponds to the reduction
rule, and lines 5--8 correspond to the extension rule. As we use lazy
lists / streams, a substitution \(\sigma_ i\) is only calculated if
proof search failed for all \(\sigma _j\) with \(j < i\).

\begin{lstlisting}[language=Haskell, caption=Lazy list implementation of clausal proof search., label=lst:lazycop]
prove [] path lim sub = [sub]
prove (lit : cla) path lim sub =
  let
    reductions = mapMaybe (unify sub (negate lit)) path
    extensions = unifyDB sub lit & concatMap
      (\ (sub1, cla1) ->
        if lim <= 0 then []
        else prove cla1 (lit : path) (lim - 1) sub1)
  in concatMap (prove cla path lim) (reductions ++ extensions)
\end{lstlisting}

\subsubsection{Continuations}\label{continuations}

Continuation passing style (CPS) allows the implementation of algorithms
with complicated control flow in functional languages
\citep{plotkin1975-continuation}. Our CPS implementation shown in
\autoref{lst:conticop} defines three functions that call each other
mutually: \lstinline!prove! starts the proof search, \lstinline!reduce!
performs reduction steps and \lstinline!extend! performs extension
steps. Two continuations are passed to the \lstinline!prove! function:
One function \lstinline!alt! to be called in case the proof search has
hit a dead end and needs to backtrack to an alternative, and one
function \lstinline!rem! to be called when a branch has been closed to
process the remaining open branches.

\begin{lstlisting}[language=Haskell, caption=CPS implementation of clausal proof search., label=lst:conticop]
prove [] path lim sub alt rem = rem (sub, alt)
prove (lit : cla) path lim sub alt rem = reduce path where
  reduce (plit : path) =
    let alt1 () = reduce path
    in case unify sub (negate lit) plit of
      Just sub1 -> prove cla path lim sub1 alt1 rem
      Nothing -> alt1 ()
  reduce [] = extend (unifyDB sub (negate lit))

  extend ((sub1, cla1) : contras) =
    let rem1 (sub, alt) = prove cla path lim sub alt rem
        alt1 () = extend contras
    in
      if lim <= 0 then alt1 ()
      else prove cla1 (lit : path) (lim - 1) sub1 alt1 rem1
  extend [] = alt ()
\end{lstlisting}

\subsubsection{Stacks}\label{stacks}

We considered an implementation based on stacks. There, the main prove
function has the same arguments as the \lstinline!prove! function of the
stream-based implementation, plus a stack. This stack contains tuples
with information about clauses that still have to be processed, together
with the depth at which the clauses have been put onto the stack. Once
the current clause has been completely refuted, the next tuple is popped
from the stack and the clause in the tuple is processed.

\subsection{Unification}\label{unification}

Unification is one of the most time consuming parts of proof search,
therefore it is crucial to represent data, including substitutions, in a
way that allows efficient unification.

The simplest approach to represent substitutions is to use association
lists from variables to terms. This is done e.g.~in the HOL Light
implementation of MESON \citep{harrison1996-meson}. However, as variable
lookup is linear in the number of bound variables, this approach does
not scale well. An improvement over this is to use tree-based maps, used
for example by Metis \citep{hurd2003-metis}. Both solutions however
incur a significant overhead in tableaux proof search, where a single
large substitution is needed.

In functional languages with efficient support for arrays (e.g.~the ML
language family, used in many proof systems), it is more efficient to
store the substitution in a single global mutable array. As variables
can be represented by positive integers, the \(n\)th array element
contains the term bound to the variable \(n\). By keeping a stack of
variables bound in each prover state, it is also possible to backtrack
efficiently: variables removed from the top of the stack are removed
from the global array. This way, backtracking can be done as if the
substitution was contained in a purely functional data structure,
however allowing for more efficient unification.

\subsection{Clause Processing Order}\label{beta-order}

Proof search processes clauses and matrices in a certain order \(<\),
such that for any elements \(a, b\) of the same clause or matrix, \(a\)
is processed before \(b\) iff \(a < b\). The order \(<\) is usually
derived from the structure of the formula obtained in
\autoref{preprocessing}. For nonclausal proof search, we have evaluated
different ways to order \(\beta\)-clauses: The original nanoCoP
processes the \(\beta\)-clause of a clause \(C\) w.r.t. a literal \(L\)
(see \autoref{def:ebclause}) using the order \(< _L\), where \(a <_L b\)
iff \(a\) contains \(L\) or \(a < b\). The reconstruction of proofs
created in this order requires some postprocessing; this motivated our
usage of the regular \(<\) for \(\beta\)-clauses.

Given an order, we can write sets as ordered sequences
\([X_1, \dots, X_n]\), where for all \(i < n\), \(X_i < X_{i+1}\).
Clauses and matrices can thus be shown as horizontal respectively
vertical sequences. We use this notation for an example to illustrate
the influence of the ordering.

\newcommand{\ncMatOneExtNew}[1]{
\left[\begin{array}{c}
M_1 \\
\left[\begin{array}{cc}
C_1 &
\left[\begin{array}{c}
#1
\end{array}\right]
\end{array}\right]
\end{array}\right]
}

\newcommand{\ncMatOneExtClassic}[1]{
\left[\begin{array}{c}
\left[\begin{array}{cc}
\left[\begin{array}{c}
#1
\end{array}\right]
& C_1
\end{array}\right]
\\ M_1
\end{array}\right]
}

\begin{example}Let \(M\) contain a clause
\[C = \ncMatOneExtNew{L \\ M_2 \\ M_3}.\] Then the \(\beta\)-clauses of
\(C\) w.r.t. \(L\) ordered by \(<\) and \(<_L\) are \(\beta _<\) and
\(\beta _{<_L}\): \[\beta _<     = \ncMatOneExtNew{M_2 \\ M_3} \qquad
  \beta _{<_L} = \ncMatOneExtClassic{M_2 \\ M_3}.\] One can see that
when using \(\beta _{<_L}\), the neighbours \(M_2\) and \(M_3\) of \(L\)
are processed first, unlike when using \(\beta _<\).\end{example}

\subsection{Evaluation}\label{funconn-eval}

We evaluate our work on several first-order problem datasets:

\begin{itemize}
\tightlist
\item
  TPTP \citep{sutcliffe2009-tptp} is a large benchmark for automated
  theorem provers. It became standard due to its usage in the yearly
  CASC competion \citep{sutcliffe2016-j8}. The contained problems are
  based on different logics and originate from various domains. In our
  evaluation we use the nonclausal first-order problems of TPTP~6.3.0.
\item
  MPTP2078 \citep{alama2014-premsel} contains 2078 problems exported
  from the Mizar Mathematical Library. This dataset is particularly
  suited for symbolic machine learning since symbols are shared between
  problems. It comes in the two flavours ``bushy'' and ``chainy'': In
  the ``chainy'' dataset, every problem contains all facts stated before
  the problem, whereas in the ``bushy'' dataset, every problem contains
  only the premises required in Mizar to prove that problem.
\item
  Miz40 contains the problems from the Mizar library for which at least
  one ATP proof has been found using one of the 14 combinations of
  provers and premise selection methods considered in
  \citep{kaliszyk2015-mizar40}. The problems are translated to untyped
  first-order logic using the MPTP infrastructure
  \citep{urban2004-mptp}. Symbol names are also used consistently in
  this dataset, and the problems are minimised using ATP-based
  pseudo-minimisation, i.e., re-running the ATP only with the set of
  proof-needed axioms until this set no longer becomes smaller. This
  typically leads to even better axiom pruning and ATP-easier problems
  than in the Mizar-based pruning used for the ``bushy'' version above.
\item
  HOL Light: We translate theorems proven in HOL Light to FOL, following
  a similar procedure as \citep{kaliszyk2014-holyhammer}. We export
  top-level theorems as well as theorems proven by the \lstinline!MESON!
  tactic.\footnote{As part of exporting theorems solved by MESON, we
    perform some of the original MESON preprocessing, such as
    propositional simplification, Skolemisation, fixing of function
    arities and so on. This preprocessing may solve the problem, in
    which case we do not export the problem at hand.} We consider the
  theorems proven in the core of HOL Light (HL) as well as those proven
  by the Flyspeck project (FS), which finished in 2014 a formal proof of
  the Kepler conjecture \citep{hales2017-kepler}.
\end{itemize}

\begin{table}

\caption{Evaluation datasets and number of contained first-order
problems. \label{fof-datasets}}

\begin{tabular}{lrrrrrrr}

\toprule

Dataset & TPTP & MPTP & Miz40 & HL-top & HL-meson & FS-top & FS-meson\tabularnewline

\midrule

Problems & 7492 & 2078 & 32524 & 2498 & 1108 & 27111 & 39979\tabularnewline

\bottomrule

\end{tabular}

\end{table}

We use a 48-core server with AMD Opteron 6174 2.2GHz CPUs, 320 GB RAM,
and 0.5 MB L2 cache per CPU. Each problem is always assigned one CPU. As
strategy scheduling is not a focus of this work, we evaluate all provers
with disabled strategy scheduling.

We evaluated several prover configurations in \autoref{tab:table1}. As
state of the art, we used the ATPs Vampire 4.0
\citep{kovacs2013-vampire} and E 2.0 \citep{schulz2013-e}, which
performed best in the first-order category of CASC-J8
\citep{sutcliffe2016-j8}. Vampire and E are written in low-level
languages (C respectively C++), implement the superposition calculus,
and perform premise selection with SInE \citep{hoder2011-sine}.
Furthermore, Vampire integrates several SAT solvers
\citep{biere2014-satvampire}, and E automatically determines proof
search settings for a given problem. We ran E with
\lstinline!--auto --auto-schedule! and Vampire with
\lstinline!--mode casc!. In addition, we evaluated the ATP Metis
\citep{hurd2003-metis}: It implements the ordered paramodulation
calculus (having inference rules for equality just like the
superposition calculus), but is considerably smaller than Vampire and E
and is implemented in a functional language, making it more comparable
to our work.

We implemented functional-style versions of leanCoP 2.1 and nanoCoP 1.0
in the functional programming language OCaml, using the techniques
introduced such as efficient control flow (\autoref{proof-search}),
array-based substitutions (\autoref{unification}), alternative clause
processing orders (\autoref{beta-order}), and consistent Skolemisation
(\autoref{skolemisation}). Our functional OCaml implementations are
fleanCoP and fnanoCoP, whereas the original Prolog versions are pleanCoP
and pnanoCoP. The Prolog versions were run with ECLiPSe 5.10. A prover
configuration containing ``\(+\)x'' or ``\(-\)x'' means that feature x
was enabled respectively disabled. ``cut'' denotes restricted
backtracking, ``conj'' stands for conjecture-directed search, and
\(\beta _{<_L}\) refers to the default \(\beta\)-clause ordering shown
in \autoref{beta-order}. leanCoP was evaluated without definitional
clausification, see \autoref{matlitdb}. The OCaml implementations use
streams to control backtracking (see \autoref{lazylist}) and arrays as
substitutions.

\begin{table}

\caption{Comparison of provers without machine learning.
\label{tab:table1}}

\begin{tabular}{lrrrrrr}

\toprule

Prover & TPTP & Bushy & Chainy & Miz40 & FS-top & FS-meson\tabularnewline

\midrule

Vampire & 4404 & 1253 & 656 & 30341 & 6358 & 39760\tabularnewline
E & 3664 & 1167 & 287 & 26003 & 7382 & 39740\tabularnewline
Metis & 1376 & 500 & 75 & 18519 & 3537 & 38625\tabularnewline
\midrule fleanCoP\(+\)cut\(+\)conj & 1859 & 670 & 289 & 12204 & 3980 & 35738\tabularnewline
fleanCoP\(+\)cut\(-\)conj & 1782 & 598 & 244 & 11796 & 3520 & 30668\tabularnewline
fleanCoP\(-\)cut\(+\)conj & 1617 & 499 & 192 & 7826 & 3849 & 35204\tabularnewline
fleanCoP\(-\)cut\(-\)conj & 1534 & 514 & 164 & 11115 & 3492 & 36334\tabularnewline
pleanCoP\(+\)cut\(+\)conj & 1673 & 606 & 182 & 11243 & 3664 & 35234\tabularnewline
pleanCoP\(+\)cut\(-\)conj & 1621 & 548 & 153 & 11227 & 3305 & 30416\tabularnewline
pleanCoP\(-\)cut\(+\)conj & 1428 & 453 & 143 & 7287 & 3671 & 34437\tabularnewline
pleanCoP\(-\)cut\(-\)conj & 1374 & 460 & 123 & 10442 & 3415 & 35499\tabularnewline
fnanoCoP\(+\)cut & 1724 & 511 & 192 & 12332 & 3178 & 30327\tabularnewline
fnanoCoP\(+\)cut\(-\beta_{<_L}\) & 1776 & 547 & 233 & 11197 & 3182 & 30216\tabularnewline
fnanoCoP\(-\)cut & 1567 & 542 & 151 & 13316 & 1993 & 37938\tabularnewline
fnanoCoP\(-\)cut\(-\beta_{<_L}\) & 1559 & 541 & 152 & 13173 & 1991 & 37923\tabularnewline
pnanoCoP\(+\)cut & 1585 & 480 & 112 & 11921 & 2970 & 30272\tabularnewline
pnanoCoP\(-\)cut & 1485 & 510 & 126 & 12943 & 1986 & 38015\tabularnewline

\bottomrule

\end{tabular}

\end{table}

The results are shown in \autoref{tab:table1}: The OCaml versions
clearly outperform the Prolog versions in almost all cases. The most
impressive result is achieved by leanCoP\(+\)cut\(+\)conj on the chainy
dataset: The OCaml version proves 58.8\% more problems than its Prolog
counterpart, thus even passing E. \(\beta_{<_L}\) seems to have an
effect mostly when cut is enabled. However, the result depends greatly
on the dataset: On the chainy dataset, disabling \(\beta _{<_L}\) solves
21.3\% more problems, but on the Miz40 dataset, it solves 8.8\% less.

We evaluated different proof search implementation styles in
\autoref{tab:clausal-impl} and \autoref{tab:nonclausal-impl}. Here,
inferences denote the number of successful unifications performed by
some prover on all problems within 10 seconds timeout. This metric is
not available for the Prolog versions, as they do not print the number
of inferences performed so far when prematurely terminated.

To measure the impact of the substitution structure, we evaluated the
best-performing implementation, i.e.~the stack-based one, using a
list-based substitution instead of an array-based substitution, see
\autoref{tab:clausal-impl}. This decreased the number of inferences by
50\%, showing that the performance of the substitution structure is
crucial for fast proof search.

Unless noted otherwise, we will use the stream-based implementation with
array-based substitution in the remainder of this paper.

\begin{table}

\caption{Impact of implementation on efficiency of clausal proof search
on the bushy MPTP2078 dataset with 10s timeout, restricted backtracking
(\(+\)cut), no definitional CNF, and conjecture-directed search
(\(+\)conj). \label{tab:clausal-impl}}

\begin{tabular}{lrr}

\toprule

Implementation & Solved & Inferences\tabularnewline

\midrule

Prolog & 606 & -\tabularnewline
Lazy list & 639 & 878199349\tabularnewline
Stack (list substitution) & 648 & 1253862954\tabularnewline
Stream & 670 & 1702827032\tabularnewline
Continuation & 681 & 2200272406\tabularnewline
Stack & 681 & 2490100879\tabularnewline

\bottomrule

\end{tabular}

\end{table}

\begin{table}

\caption{Impact of implementation on efficiency of nonclausal proof
search on the bushy MPTP2078 dataset with 10s timeout and restricted
backtracking (\(+\)cut). \label{tab:nonclausal-impl}}

\begin{tabular}{lrr}

\toprule

Implementation & Solved & Inferences\tabularnewline

\midrule

Prolog & 480 & -\tabularnewline
Lazy list & 504 & 374849495\tabularnewline
Streams & 511 & 495368962\tabularnewline

\bottomrule

\end{tabular}

\end{table}

\section{Consistent Skolemisation}\label{skolemisation}

First-order Skolemisation introduces new function symbols. For machine
learning algorithms, it is beneficial to introduce names
\emph{consistently} across problems, meaning that Skolem terms
originating from the same axiom in two different problems should be
syntactically equivalent. Consistent Skolemisation methods have been
studied in the context of the \(\delta\)-rule in tableaux methods, e.g.
\citep{beckert1993-delta}. \citep{giese1999-epsilon} pointed out that
the Skolem terms introduced may lead to rather large formulae, which can
be solved by structure sharing. However, in our setting, structure
sharing across different problems is not possible, which makes it
necessary to find different approaches. In previous work
\citep{kaliszyk2015-femalecop}, consistent Skolemisation was part of the
clausification procedure. The implementation of nonclausal proof search
motivated a more general consistent Skolemisation method. To recognise
the same Skolem term across different problems, it is necessary to
capture by the Skolem term only the part of the formula that defines the
existential variable. For this, we propose a consistent Skolemisation
method based on \(\epsilon\)-terms.

Let us first introduce the setting for this section: Let \(\Delta\) be a
rectified formula in negation normal form that is to be skolemised.
Furthermore, let the size of a formula \(F\) be the length of the string
representation of \(F\), and denote it by \(|F|\).

\begin{definition}[Skolemisation]The Skolemisation of a formula
\(\Delta\) yields a formula equisatisfiable to \(\Delta\), not
containing any existential quantifiers. To this end, Skolemisation
replaces any subterm in \(\Delta\) of the shape \(\exists x. F\) by
\(F[t/x]\), where \(t\) is called the Skolem term for
\(x\).\end{definition}

Replacing existentially quantified variables with \(\epsilon\)-terms
using the defining property of \(\epsilon\)-terms
\[\exists x. P(x) \leftrightarrow P(\epsilon x. P(x)),\] we obtain a
formula equivalent to the original one \citep{hilbert1939-grundlagen}.
However, recursively replacing existential quantifiers by
\(\epsilon\)-terms can lead to an exponential size of the skolemised
formula. To show this, we are going to define two kinds of
\(\epsilon\)-Skolemisation and show that they both produce exponentially
large output. In particular, the blowup is caused by the introduction of
new Skolem names that contain other Skolem names.

\begin{definition}Let \(F\) a subformula of \(\Delta\) such that \(F\)
is of the shape \(\exists x. G\). A \emph{naive
\(\epsilon\)-Skolemisation step} replaces \(F\) in \(\Delta\) by
\(G[(\epsilon x. G) / x]\).\end{definition}

Naive \(\epsilon\)-Skolemisation (N\(\epsilon\)S) of a formula
\(\Delta\) repeatedly applies naive \(\epsilon\)-Skolemisation steps
until the formula does not contain any more existential quantifiers. To
fix the order of Skolemisation steps, let outside-in N\(\epsilon\)S
replace subformulas only if they are not subformulas of an existential
quantification, and inside-out N\(\epsilon\)S replace subformulas only
if they do not have subformulas containing existential quantifiers. We
now give an example for which both outside-in and inside-out
N\(\epsilon\)S produce exponentially large skolemised
formulas.\footnote{Structure sharing would avoid the exponential blowup.}

\begin{example}Let \(\phi _n\) be a formula recursively defined by
\(\phi _0 = P(x_0, x_0)\) and
\(\phi _{n+1} = \exists x_n. (P(x_{n+1}, x_{n+1}) \rightarrow \phi _n)\).\end{example}

\begin{lemma}Inside-out N\(\epsilon\)S of \(\phi _n\) produces a formula
exponential in \(n\).\end{lemma}

\begin{proof}Denote the inside-out N\(\epsilon\)S of \(\phi _n\) as
\(sk(\phi _n)\). Then \(sk(\phi _0) = P(x_0, x_0)\) and
\(sk(\phi _{n+1}) = P(x_{n+1}, x_{n+1}) \rightarrow sk(\phi _n)[t_n / x_n]\),
where
\(t_n = \epsilon x_n. (P(x_{n+1}, x_{n+1}) \rightarrow sk(\phi _n))\) is
the Skolem term corresponding to \(x_n\). For every \(n\),
\(sk(\phi _n)\) contains at least two occurrences of \(x_n\), and the
Skolem term \(t_n\) corresponding to \(x_n\) is larger than
\(|sk(\phi _n)|\). Therefore, for every \(n\),
\(|sk(\phi _{n+1})| > 2 |sk(\phi _n)|\).\end{proof}

\begin{lemma}Outside-in N\(\epsilon\)S of \(\phi _n\) produces a formula
exponential in \(n\).\end{lemma}

\begin{proof}Let \(\Delta = \forall x_m. \phi _m\) be the formula to be
skolemised. Then the Skolem terms corresponding to \(x_n\) can be given
by \[s_n =
  \begin{cases}
x_m & \text{if } n = m \\
\epsilon x_n. P(s_{n+1}, s_{n+1}) \rightarrow \phi _n & \text{otherwise}
  \end{cases}\] For any \(n < m\), because every Skolem term \(s_n\)
contains two occurrences of \(s_{n+1}\), we have that
\(|s_n| > 2|s_{n+1}|\). As the base case \(s_m\) is greater than zero,
we have that \(|s_0|\) is exponential in \(m\).\end{proof}

The example above motivates a new consistent Skolemisation method that
produces quadratic output and is also applicable to nonclausal search.
For this, let us define some notation first: \(\FVar(F)\) denotes the
free variables of \(F\), and \(\FVar_\forall(F)\) respectively
\(\FVar_\exists(F)\) denote the free variables in a subformula \(F\) of
\(\Delta\) that are universally respectively existentially bound in
\(\Delta\). \(\Delta^x\) is the subformula of \(\Delta\) that binds the
variable \(x\), i.e.~if \(\Delta\) has a subformula \(\exists x. F\),
then \(\Delta^x\) is \(\exists x. F\). Furthermore, the transitive
universal respectively existential free variables of \(F\) are denoted
by
\[\FVar_Q^*(F) = \FVar_Q(F) \cup \bigcup _{x \in \FVar_\exists(F)} \FVar_Q^*(\Delta^x),\]
where \(Q \in \{\forall, \exists\}\).

\begin{definition}[Consistent
\(\epsilon\)-Skolemisation]\label{def:ces}Let \(x\) be an existentially
quantified variable in \(\Delta\). The \(\epsilon\)-defining formula of
\(\Delta^x\) is
\(\max \{ \Delta^y \mid y \in \FVar_\exists^*(\Delta^x) \}\).\footnote{The
  maximum is computed with respect to the size of the formula, that is,
  it returns the largest formula in the set.} The consistent
\(\epsilon\)-Skolem term for \(x\) is \(\epsilon x. F\), where \(F\) is
the \(\epsilon\)-defining formula of \(\Delta^x\) with all quantifiers
for \(\FVar_\forall^*(\Delta^x)\) and \(x\) removed. Consistent
\(\epsilon\)-Skolemisation of a formula \(\Delta\) replaces any
subformula of \(\Delta\) of the shape \(\exists x. F\) by \(F[t / x]\),
where \(t\) is the consistent \(\epsilon\)-Skolem term for
\(x\).\end{definition}

\begin{example}Let \[\Delta =
\forall x_1 \exists y_1. (P(x_1, y_1) \rightarrow
   (\forall x_2 \exists y_2. P(x_2, y_2)) \land
   (\forall x_3 \exists y_3. Q(x_3, y_3, y_1))).\] The
\(\epsilon\)-defining formula of \(\Delta^{y_1}\) and \(\Delta^{y_3}\)
is \(\Delta^{y_1}\), and of \(\Delta^{y_2}\), it is \(\Delta^{y_2}\).
The consistent \(\epsilon\)-Skolem term for \(y_n\) is \(s_n\), where
\small
 \[\begin{aligned}
  s_1 &= \epsilon y_1. P(x_1, y_1) \rightarrow
(\forall x_2 \exists y_2. P(x_2, y_2)) \land
(\forall x_3 \exists y_3. Q(x_3, y_3, y_1)) &
\FVar(s_1) &= \{x_1\} \\
  s_2 &= \epsilon y_2. P(x_2, y_2) &
\FVar(s_2) &= \{x_2\} \\
  s_3 &= \epsilon y_3. \exists y_1. P(x_1, y_1) \rightarrow
(\forall x_2 \exists y_2. P(x_2, y_2)) \land
Q(x_3, y_3, y_1) &
\FVar(s_3) &= \{x_1, x_3\}
  \end{aligned}\]\end{example}

\begin{theorem}Consistent \(\epsilon\)-Skolemisation of a formula
\(\Delta\) yields a formula equivalent to \(\Delta\).\end{theorem}

\begin{proof}Let us consider an arbitrary existentially quantified
variable \(x\) in \(\Delta\). Let \(D\) be the \(\epsilon\)-defining
formula of \(\Delta^x\). We can move all universal quantifiers
corresponding to \(\FVar_\forall^*(\Delta^x)\) in front of \(D\),
yielding \(\forall \vec y_1. D_1\). Note that
\(\FVar(D_1) = \FVar_\forall^*(\Delta^x)\) and
\(\vec y_1 = \FVar_\forall^*(\Delta^x) \setminus \FVar_\forall(D)\).
Furthermore, we can move the existential quantifier for \(x\) in front
of \(D_1\), resulting in \(\forall \vec y_1 \exists x. D_2\). Now, we
can use the defining property of \(\epsilon\)-terms to obtain
\(\forall \vec y_1. D_2[t / x]\), where \(t = \epsilon x. D_2\).
Finally, we can move back the quantifiers from \(\vec y_1\) to their
original places to yield \(D_3\). (\(D_3\) could have been equally
obtained by removing the quantifier \(\exists x\) from \(D\) and
replacing \(x\) by \(\epsilon x. D_2\).) As all free variables of every
such \(\epsilon\)-term \(t\) are universally quantified, we can execute
the operations above to replace in arbitrary order every existentially
quantified variable with its corresponding \(\epsilon\)-term. As all the
operations preserve equivalence of the formula, this yields a formula
equivalent to \(\Delta\), and we can see that it is exactly the outcome
of consistent \(\epsilon\)-Skolemisation.\end{proof}

\begin{lemma}Consistent \(\epsilon\)-Skolemisation of \(\Delta\) yields
a formula of size smaller than \(|\Delta|^2\).\end{lemma}

\begin{proof}The maximal size of a consistent \(\epsilon\)-Skolem term
is \(|\Delta|\), and as there are less than \(|\Delta|\) occurrences of
existentially quantified variables in \(\Delta\), replacing them yields
a formula of size smaller than \(|\Delta|^2\).\end{proof}

To obtain a first-order formula from consistent
\(\epsilon\)-Skolemisation, we replace every consistent
\(\epsilon\)-Skolem term of the shape \(\epsilon x. F\) with
\(f_{[\epsilon x. F]}(\vec y)\), where \(\vec y\) is
\(\FVar(\epsilon x. F)\) (and thus also \(\FVar_\forall^*(\Delta^x)\)).
Here, \([t]\) denotes a normalisation of \(t\) such that for any \(t_1\)
and \(t_2\), if \(t_1\) is \(\alpha\)-equivalent to \(t_2\), then
\([t_1] = [t_2]\). We use this normalisation to recognise equivalent
Skolem terms even if their variables are differently named. Due to the
introduction of new function symbols, this yields a formula
equisatisfiable to \(\Delta\).

The following corollary states under which conditions variables that are
existential bound at different locations will be consistently skolemised
to the same function symbol. Note that this result holds across
different formulas and problems.

\begin{corollary}[Consistency]Two existential variables are mapped to
the same first-order function symbol iff their corresponding
\(\epsilon\)-Skolem terms are \(\alpha\)-equivalent.\end{corollary}

\begin{proof}Follows from \autoref{def:ces}.\end{proof}

\section{Naive Bayesian Internal Guidance}\label{bayesguidance}

The order in which extension steps are tried can have a significant
effect on the performance of proof search. In this section, we propose
the use of Naive Bayesian probability to \emph{guide} the use of
extension steps based on an intermediate proof state.

We generally call methods that implement guidance inside ATPs
\emph{internal guidance} methods. In the particular case of machine
learning guidance, such methods are historically motivated by the
relative success of the \emph{external guidance} methods used mainly for
premise selection outside of the core ATP systems
\citep{blanchette2016-qed, urban2008-malarea, kaliszyk2014-holyhammer}.
Internal guidance methods aim to estimate the utility of actions
according to the system's knowledge of the world and previous
experiences. In our setting, a positive experience is when in a certain
tableau branch, an extension step used a contrapositive that contributed
to the final proof. We use this information to prefer contrapositives in
branches that are similar to branches where the contrapositives were
previously useful.

To measure the similarity between tableau branches, we characterise them
by \emph{features} \citep{kaliszyk2015-features}, which we explain in
\autoref{features}. In \autoref{naive-bayes}, we then calculate the
utility of a contrapositive in the current branch, given knowledge about
its utility in previous proofs. In \autoref{implementations}, we
motivate the integration of machine learning methods in the prover and
introduce the prover FEMaLeCoP, which we evaluate in \autoref{nb-eval}.

\subsection{Tableau Branch Characterisation}\label{features}

The words of the connection tableaux calculus
\(\langle C, M, Path \rangle\) correspond to a set of tableau branches
sharing the active \(Path\). Therefore, to characterise a branch, we use
as its \emph{features} the set of symbols occurring in the active path.
We weigh the symbols by the number of times they appeared in all
problems, giving higher weight to rarer symbols via \emph{inverse
document frequency} \citep{jones1973-idf}, as well as by the distance
between the current depth and the depth the symbols where put onto the
path, giving higher weight to symbols more recently processed.

\subsection{Naive Bayes}\label{naive-bayes}

Given a set of contrapositives that are applicable in a tableau branch,
we wish to obtain an ordering of the contrapositives such that trying
the contrapositives in the given order minimises the time spent to find
a proof. In this subsection, we show how to order the set of applicable
contrapositives by a formula \emph{\(\nb\)} that is based on Naive
Bayesian probability, as used for premise selection
\citep{kaliszyk2015-mizar40}.

First, let us denote the knowledge about the usage of contrapositives in
previous proofs by \(F(l_i)\), which is the multiset of sets of features
characterising the tableau branches in which the usage of \(l_i\)
contributed to the final proof. \(|F(l_i)|\) is the total number of
times that \(l_i\) was used in previous proofs.

\begin{example}\(F(l_1) = \left\{\left\{f_1, f_2\right\}, \left\{f_2, f_3\right\}\right\}\)
means that the contrapositive \(l_1\) was used twice in previous proofs;
once in a proof state characterised by the features \(f_1\) and \(f_2\),
and once when features \(f_2\) and \(f_3\) were present.\end{example}

Let \(P(l_i, \vec f)\) denote the probability that a contrapositive
\(l_i\) from a set \(\vec l\) of applicable contrapositives is useful in
a tableau branch characterised by features \(\vec f\). Using Bayes'
theorem together with the (naive) assumption that features are
statistically independent, we derive
\[P(l_i \mid \vec f) = \frac{P(l_i) P(\vec f \mid l_i)}{P(\vec f)} =
  \frac{P(l_i)}{P(\vec f)} \prod _ {f_j \in \vec f} P(f_j \mid l_i).\]
To increase numerical stability, we calculate the logarithm of the
probability
\[\ln P(l_i \mid \vec f) = \ln P(l_i) - \ln P(\vec f) + \sum _ {f_j \in \vec f} \ln P(f_j \mid l_i).\]
In the final formula \(\nb(l_i, \vec f)\) to rank contrapositives, we
modify \(\ln P(l_i \mid \vec f)\) as follows:

\begin{itemize}
\tightlist
\item
  We add a term to disadvantage features not present in \(\vec f\) that
  occurred in previous situations where the contrapositive was useful.
\item
  We weigh the probability of any feature \(f\) by its inverse document
  frequency \(i(f)\) to give more weight to rare features.
\item
  We drop the term \(\ln P(\vec f)\), as we compare only values for
  fixed features \(\vec f\).
\item
  We weigh the individual parts of the sum with constants \(\sigma_1\),
  \(\sigma_2\) and \(\sigma_3\).
\end{itemize}

The resulting formula is \[\nb(l_i, \vec f) =
\sigma_1 \ln P(l_i) +
\sigma_2 \sum _ {f_j \in \vec f} i(f_j) \ln P(f_j \mid l_i) +
\sigma_3 \sum _ {f_j \in \bigcup F(l_i) \setminus \vec f} i(f_j) \ln (1 - P(f_j \mid l_i))\]

We are now going to describe how to calculate \(P(l_i)\) as well as
\(P(f_j \mid l_i)\). First, we calculate the unconditional
contrapositive probability as
\[P(l_i) = \frac{\left| F(l_i) \right|}{\sum _ {l_j \in \vec l} \left| F(l_j) \right|}.\]
In practice, as the denominator of the fraction is the same for all
\(l_i\), we drop it, similarly to \(P(\vec f)\) above. To obtain the
conditional feature probability, we distinguish whether a feature
\(f_j\) already appeared in conjunction with a contrapositive \(l_i\).
If so, then its probability is the ratio of times \(f_j\) appeared when
\(l_i\) was used and all times that \(l_i\) was used. Otherwise, the
probability is estimated to be a minimal constant probability \(\mu\):
\[P(f_j \mid l_i) =
\begin{cases}
  \sum _ {\vec{f'} \in F(l_i)} 1 _ {\vec{f'}}(f_j) / \left| F(l_i) \right|
  & \text{if } \exists \vec{f'} \in F(l_i). f_j \in \vec{f'} \\
  \mu & \text{otherwise}
\end{cases}\] Here, \(1_A(x)\) denotes the indicator function that
returns 1 if \(x \in A\) and 0 otherwise.

\subsection{Implementations}\label{implementations}

The \emph{Machine Learning Connection Prover} (MaLeCoP) was the first
leanCoP-based system to explore the feasibility of machine-learnt
internal guidance \citep{urban2011-malecop}. MaLeCoP relies on an
external machine learning framework (using by default the SNoW system
\citep{carlson1999-snow}), providing several machine learning
algorithms, namely Naive Bayes and shallow neural networks based on
perceptrons or winnow cells. During proof search, MaLeCoP sends features
of its current branch to the framework, which orders the proof steps
applicable in the current branch by their expected utility. The usage of
a general framework eases experiments with different methods, but the
prediction speed of MaLeCoP's underlying advisor system together with
the communication overhead was several orders of magnitude lower than
the raw inference speed of leanCoP. This was to some extent countered by
fast query caching mechanisms and a number of strategies trading the
machine-learnt advice for raw speed, yet the real-time performance of
the system remained relatively low.

This motivated the creation of the \emph{Fairly Efficient Machine
Learning Connection Prover} (FEMaLeCoP), which improved speed by
integrating a fast and optimised Naive Bayesian classifier as shown in
\autoref{naive-bayes} into the prover \citep{kaliszyk2015-femalecop}.
Naive Bayes was chosen because learning data can be easily filtered for
the current problem, making the calculation of Naive Bayesian
probabilities for a given branch efficient for each applicable
contrapositive. FEMaLeCoP efficiently calculates the Bayesian
probabilities of a given set of contrapositives by saving contrapositive
statistics directly in the literal database. Performance is further
improved as branch features are not fully recalculated in every new
branch, but updated from the previous branch.

\subsection{Evaluation}\label{nb-eval}

The evaluation involves generation of training data with leanCoP,
followed by running FEMaLeCoP with the training data on the same
problems. We run both leanCoP and FEMaLeCoP on the bushy MPTP2078
dataset with a timeout of 60s, nondefinitional clausification,
conjecture-directed search and restricted backtracking. Both leanCoP and
FEMaLeCoP considered in this evaluation are implemented in OCaml using
continuation passing style, see \autoref{continuations}.

leanCoP proves 574 problems.\footnote{The leanCoP version evaluated here
  uses a different clause order than the versions in
  \autoref{proof-search}, which explains the different baseline
  performance.} From the resulting proofs, the information is extracted
which contrapositive contributed in which tableau branch. This
information is combined for all proofs to a format that allows efficient
retrieval of learning data for given contrapositives.

With the training data generated from the leanCoP proofs, we run
FEMaLeCoP on the same problems as leanCoP. While the inference rate
drops by about 40\%, FEMaLeCoP proves 640 problems. The union of leanCoP
and FEMaLeCoP proves 664 problems, adding 90 problems (15.7\%) to the
problems solved by leanCoP. A more thorough evaluation of Naive Bayesian
machine learning integrated with other internal guidance components will
be performed in the next section.

\section{Monte Carlo Proof Search}\label{monte-carlo-proof-search}

In this section, we describe how to expand a proof search tree using
Monte Carlo Tree Search (MCTS).

For an intuition of the relationship between different proof search
strategies, see \autoref{fig:leancop-montecop}: Iterative deepening
considers all potential proof trees of a certain depth before
considering trees of higher depth. Restricted backtracking uniformly
discards a set of potential proof trees. MCTS allows for a more
fine-grained proof search, searching different regions of the search
space more profoundly than others, based on heuristics.

We introduce MCTS and propose a set of heuristics adapted to proof
search. Then, we show an implementation of the method, closing with an
evaluation.

\begin{figure}
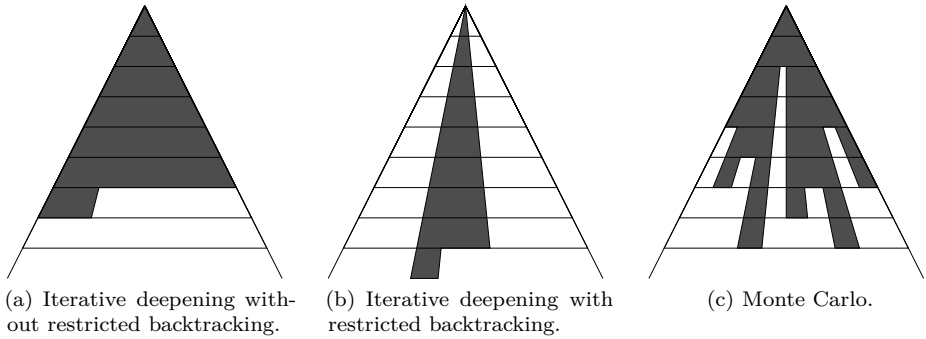

  \subfloat[Iterative deepening without restricted backtracking.]{%
    \resizebox{0.3\columnwidth}{!}{\begin{tikzpicture}

\filldraw[fill=black!70!white,draw=black]
  (0,0) -- ({-0.5*7}, -7) -- ({-0.25*7}, -7) -- ({-0.25*6}, -6) -- ({0.5*6}, -6) -- cycle;

\input{tikz/iterdeep-common}

\end{tikzpicture}}
  }
  \hfill
  \subfloat[Iterative deepening with restricted backtracking.]{%
    \resizebox{0.3\columnwidth}{!}{\begin{tikzpicture}

\filldraw[fill=black!70!white,draw=black]
  (0,0) -- ({-0.2*9}, -9) -- ({-0.1*9}, -9) -- ({-0.1*8}, -8) -- ({0.1*8}, -8) -- cycle;

\input{tikz/iterdeep-common}

\end{tikzpicture}}
  }
  \hfill
  \subfloat[Monte Carlo.]{%
    \resizebox{0.3\columnwidth}{!}{\begin{tikzpicture}

\filldraw[fill=black!70!white,draw=black]
  (0,0) --
  ({-0.5*4}, -4) -- ({-0.4*4}, -4) --
  ({-0.4*6}, -6) -- ({-0.3*6}, -6) --
  ({-0.3*5}, -5) -- ({-0.2*5}, -5) --
  ({-0.2*8}, -8) -- ({-0.1*8}, -8) --
  ({-0.1*2}, -2) -- ({ 0.0*2}, -2) --
  ({ 0.0*7}, -7) -- ({ 0.1*7}, -7) --
  ({ 0.1*6}, -6) -- ({ 0.2*6}, -6) --
  ({ 0.2*8}, -8) -- ({ 0.3*8}, -8) --
  ({ 0.3*4}, -4) -- ({ 0.4*4}, -4) --
  ({ 0.4*6}, -6) -- ({ 0.5*6}, -6) --
   cycle;

\input{tikz/iterdeep-common}

\end{tikzpicture}}
  }
  \caption{The two main leanCoP strategies compared with Monte Carlo proof search.}
  \label{fig:leancop-montecop}
\end{figure}

\subsection{Monte Carlo Tree Search}\label{mcts}

Monte Carlo Tree Search (MCTS) is a method to search potentially
infinite trees by sampling random tree paths (called \emph{simulations})
\citep{browne2012-survey}. The outcome of simulations is then used to
estimate the quality of tree nodes, and MCTS steers search towards nodes
with higher quality estimates.

\begin{definition}[Tree]A \emph{tree} is a tuple
\((N, n_0, \rightarrow)\), where \(N\) is a set of tree nodes,
\(n_0 \in N\) is the root node, and \(\rightarrow \; \in N \times N\) is
a cycle-free relation, i.e.~there is no \(n \in N\) such that
\(n \rightarrow^+ n\).\footnote{\(\rightarrow ^+\) is the transitive
  closure of \(\rightarrow\).} We write that \(n'\) is a child of \(n\)
iff \(n \rightarrow n'\). Every \(n \in N\) is the child of maximally
one node in \(N\).\end{definition}

We consider proof search as traversal of a (usually infinite) tree
\((N, n_0, \rightarrow)\), such that \(N\) is the set of derivations
(tableaux), \(n_0\) is a derivation that consists of some word
\(\langle C, M, Path \rangle\) corresponding to the matrix \(M\) of a
given problem, and \(n \rightarrow n'\) iff \(n'\) can be obtained from
\(n\) by applying a single calculus rule. If \(n \rightarrow n'\) by a
single application of the extension rule using the contrapositive \(c\),
then we write \(n \xrightarrow{\ext(c)} n'\). Proof search succeeds when
we find a leaf node of \(N\) that is a proof.

Let \(\rho \in N \rightarrow \mathbb{R}\) be a \emph{reward function}
that estimates the distance of an unclosed derivation in the proof
search tree from a closed derivation. Then we can use Monte Carlo Tree
Search to traverse the proof search tree, giving preference to regions
that yield higher rewards. For this, we first define Monte Carlo trees:

\begin{definition}[Monte Carlo tree]A \emph{Monte Carlo tree} \(T\) for
a tree \((N, n_0, \rightarrow)\) is a tuple
\((N_T, \rightarrow _T, \rho _T)\), where
\(\rightarrow _T \: \subseteq \: \rightarrow ^+\) and
\(\rho _T \in N \rightarrow \mathbb{R}\) is a mapping. We write that
\(n'\) is a \(T\)-child of \(n\) iff \(n \rightarrow _T n'\). The
\emph{initial Monte Carlo tree} \(T_0\) is
\((N_{T_0}, \rightarrow _{T_0}, \rho _{T_0})\) with
\(N_{T_0} = \{n_0\}\), \(\rightarrow _{T_0} = \emptyset\) and
\(\rho _{T_0}(n) = 0\) for all \(n\).\end{definition}

A single iteration of Monte Carlo Tree Search takes a Monte Carlo tree
\(T\) and returns a new tree \(T'\) as follows:\footnote{Frequently,
  MCTS is described to have a \emph{backpropagation} step that adds
  rewards to the ancestors of the newly added nodes. We omit this step,
  adapting the child selection policy instead.}

\begin{enumerate}
\def\labelenumi{\arabic{enumi}.}
\tightlist
\item
  \textbf{Selection}: A node \(n \in N_T\) with
  \(n_0 \rightarrow^* _T n\) is chosen with a \emph{child selection
  policy}, see \autoref{childselpol}.
\item
  \textbf{Simulation}: A child \(n_1\) of \(n\) is randomly chosen with
  child probability \(P(n_1 \mid n)\) to be the \emph{simulation root},
  see \autoref{childprob}. (Every tree node is chosen maximally once to
  be a simulation root, to guarantee the exploration of the tree.) From
  \(n_1\), a sequence of random transitions
  \(n_1 \rightarrow \dots \rightarrow n_s\) is performed, where for
  every \(i < s\), \(n_{i+1}\) is randomly selected with child
  probability \(P(n_{i+1} \mid n_i)\).
\item
  \textbf{Expansion}: A node \(n_e\) from
  \(n_1 \rightarrow \dots \rightarrow n_s\) is selected with the
  \emph{expansion policy}, see \autoref{exppol}. The node \(n_e\) is
  added as a child to \(n\) with reward \(\rho(n_s)\) (see
  \autoref{reward}) to yield the new tree \(T'\):
  \[N_{T'} = N_T \cup \{n_e\} \qquad
     \rightarrow _{T'} \; = \; \rightarrow _T \cup \; \{(n, n_e)\} \qquad
     \rho _{T'} = \rho _T \{n_e \mapsto \rho (n_s)\}\]
\end{enumerate}

In the next sections, we show heuristics for the child selection policy,
child probability, reward, and expansion policy.

\subsection{Child Selection Policy}\label{childselpol}

UCT (Upper Confidence Bounds for Trees) is a frequently used child
selection policy for Monte Carlo Tree Search \citep{kocsis2006-uct}. It
uses \(\visits_T(n)\), which is the number of \(T\)-descendants of
\(n\), and \(\overline \rho _T(n)\), which is the average
\(T\)-descendant reward of \(n\).
\[\visits_T(n) = |\{n' \mid n \rightarrow^+_T n' \}| \qquad \qquad
\overline \rho _T(n) =
\frac
{\sum \{ \rho _T(n') \mid n \rightarrow^*_T n' \}}
{\visits_T(n)}\] Given a node \(n\), UCT ranks every \(T\)-child \(n'\)
of \(n\) with
\[\uct(n, n') = \overline \rho _T(n') + C_p \sqrt \frac{\ln \visits_T(n)}{\visits_T(n')}.\]
Here, \(C_p\) is called the \emph{exploration constant}, where small
values of \(C_p\) prefer nodes with higher average descendant reward and
large values of \(C_p\) give more weight to nodes with fewer visits. In
the UCT formula, division by zero is expected to yield \(\infty\), so if
a node \(n\) has unvisited children, one of them will be selected by
UCT.

The UCT child selection policy \(cs_T(n)\) recursively traverses the
Monte Carlo tree \(T\) starting from the root \(n_0\). \(cs_T(n)\)
chooses the \(T\)-child of \(n\) with maximal UCT value and recurses,
unless \(n\) has no \(T\)-child, in which case \(n\) is returned:

\[cs_T(n) = \begin{cases}
cs_T\left( \argmax_{\{n' \mid n \rightarrow _T n'\}} \uct(n, n') \right)
  & \text{if } \exists n'. \; n \rightarrow _T n' \\
n & \text{otherwise}
\end{cases}\]

\subsection{Child Probability}\label{childprob}

The child probability \(P(n' \mid n)\) determines the likelihood of
choosing a child node \(n'\) of \(n\) in a simulation. We show three
different methods to calculate the child probability.

\begin{itemize}
\tightlist
\item
  The \emph{baseline probability} assigns equal probability to all
  children, i.e. \(P(n' \mid n) \propto 1\).
\item
  The \emph{open branches probability} steers proof search towards
  derivations with fewer open branches, by assigning to \(n'\) a
  probability inversely proportional to the number of open branches in
  \(n'\). Therefore, \(P(n' \mid n) \propto 1 / (1 + |b_o(n')|)\), where
  \(b_o(n)\) returns the open branches in \(n\).
\item
  The \emph{Naive Bayes probability} attributes to \(n'\) a probability
  depending the calculus rule applied to obtain \(n'\) from \(n\): In
  case the extension rule was not used, the node obtains a constant
  probability. If the extension rule was used, the formula \(\nb\)
  introduced in \autoref{naive-bayes} is used, requiring contrapositive
  statistics from previous proofs. However, as \(\nb\) is not a
  probability, we use it to rank contrapositives by the number of
  contrapositives with larger values of \(\nb\):
  \[\rank_{\nb}(n, c) = \left| \left\{
    c' \mid n \xrightarrow{\ext(c')} n', \nb(c', \vec f(n)) \geq \nb(c, \vec f(n))
    \right\} \right|,\] where \(\vec f(n)\) denotes the features of the
  derivation \(n\). Then, we assign to nodes as probability the inverse
  of the Naive Bayes rank: \[P(n' \mid n) \propto \begin{cases}
  1 / \rank_{\nb}(n, c) & \text{if } n \xrightarrow{\ext(c)} n' \\
  1 & \text{otherwise}
    \end{cases}\]
\end{itemize}

\subsection{Reward}\label{reward}

The reward heuristic estimates the likelihood of a given derivation to
be closable. This is in contrast to most prover heuristics (such as
child probability) that only compare the quality of children of the same
node. We use our reward heuristics to evaluate the last node \(n\) of a
simulation.

Several heuristics in this section require a normalisation function, for
which we use a strictly increasing function
\(\norm : [0, \infty) \to [0, 1)\) that fulfils
\(\lim_{x \to \infty} \norm(x) = 1\) and \(\norm(0) = 0\). For example,
\(\norm(x) = 1 - (x + 1)^{-1}\).

\begin{itemize}
\tightlist
\item
  The \emph{branch ratio reward} determines the reward to be the ratio
  of the number of closed branches and the total number of branches,
  i.e. \(\rho(n) = |b_c(n)| / |b(n)|\).
\item
  The \emph{branch weight reward} is based on the idea that many open
  branches with large literals are indicators of a bad proof attempt.
  Here, the size \(|l|\) of a literal is measured by the number of
  symbol occurrences in \(l\). Furthermore, the closer to the derivation
  root a literal appears, the more characteristic we consider it to be
  for the derivation. Therefore, the reward is the average of the
  inverse size of the branch leafs, where every leaf is weighted with
  the normalised depth of its branch. \[\rho(n) = \frac{1}{|b_o(n)|}
    \sum _{b \in b_o(n)} \frac{\norm(\depth(b))}{|\leaf(b)|}\]
\item
  The \emph{machine-learnt closability reward} assumes that the success
  ratio of closing a branch in previous derivations can be used to
  estimate the probability that a branch can be closed in the current
  derivation. This needs the information about attempted branches in
  previous derivations, and which of these attempts were successful. We
  say that a literal \(l\) stemming from a clause \(c\) is attempted to
  be closed during proof search when \(l\) lies on some branch. The
  attempt is successful iff proof search manages to close all branches
  going through \(l\). Given such data from previous proof searches, let
  \(p(l)\) and \(n(l)\) the number of successful respectively
  unsuccessful attempts to close \(l\). We define the
  \emph{unclosability} of a literal \(l\) as
  \(\frac{n(l)}{p(l) + n(l)}\). However, the less data we have about a
  literal, the less meaningful our statistics will be. To account for
  this, we introduce \emph{weighted unclosability}: We assume that a
  literal that never appeared in previous proof searches is most likely
  closable, i.e.~its weighted unclosability is 0. The more often a
  literal was attempted to be closed, the more its weighted
  unclosability should converge towards its (basic) unclosability.
  Therefore, we model the probability of \(l\) to be closable as
  \[P(l \closable) = 1 - \norm(p(l)+n(l)) \frac{n(l)}{p(l)+n(l)}.\]
  Finally, the closability of a derivation is the mean closability of
  all leafs of open branches of the derivation, i.e.~the final reward
  formula is
  \(\rho(n) = \sum _{b \in b_o(n)} P(\leaf(b)\, \closable) / |b_o(n)|\).
\end{itemize}

To measure the efficiency of a reward heuristic, we introduce
\emph{discrimination}: Assume that an MCTS iteration of the Monte Carlo
tree \(T\) starts a simulation from the node \(n_p\) and finds a proof.
Then the discrimination of \(T\) is the ratio of the average reward on
the Monte Carlo tree branch from the root node \(n_0\) to \(n_p\) and
the average reward of all Monte Carlo tree nodes. Formally, let the
average reward of a set of nodes \(N\) be
\(\overline \rho _T(N) = \sum \{\rho _T(n) \mid n \in N \} / |N|\). Then
the discrimination of \(T\) is \[\frac
{\overline \rho _T(\{ n \mid n_0 \to _T^* n, n \to _T^* n_p \})}
{\overline \rho _T(\{ n \mid n_0 \to _T^* n \})}\]

\subsection{Expansion Policy}\label{exppol}

The expansion policy determines which node \(n_e\) of a simulation
\(n_1 \rightarrow \dots \rightarrow n_s\) is added to the Monte Carlo
tree. We implement two different expansion policies:

\begin{itemize}
\tightlist
\item
  The \emph{default expansion policy} adds \(n_1\), i.e.~the simulation
  root, to the MC tree.
\item
  The \emph{minimal expansion policy} picks \(n_e\) to be the smallest
  of the simulation nodes w.r.t. a given norm \(|\cdot|\), such that for
  all \(i\), \(|n_e| \leq |n_i|\). If multiple \(n_e\) are admissible,
  the one with the smallest index \(e\) is picked. We consider two norms
  on nodes:

  \begin{enumerate}
  \def\labelenumi{\arabic{enumi}.}
  \tightlist
  \item
    The first norm measures the number of open branches.
  \item
    The second norm measures the sum of depths of open branches.
  \end{enumerate}
\end{itemize}

The minimal expansion policy is similar to restricted backtracking in
the sense that it restricts proof search to be resumed only from certain
states, thus resulting in an incomplete search.

\subsection{Implementation}\label{implementation}

We implemented Monte Carlo proof search (MCPS) based on leanCoP, where
leanCoP builds the search tree and MCTS chooses which regions of the
tree to search. Unlike for the traditional leanCoP, the depth of the
search tree is not limited. To guarantee nonetheless that simulations
terminate, simulations are stopped after a fixed number of simulation
steps \(s_{\max}\).

While it is possible to run MCPS from the root node until a proof is
found, we have found it to perform better when it serves as
\emph{advisor} for leanCoP. We show this in \autoref{lst:montecop},
assuming for a simpler presentation that the default expansion policy
from \autoref{exppol} is used: MCPS as performed by
\lstinline!mcps lit path sub! attempts to close the branch containing
the literals in \lstinline!path! and the literal \lstinline!lit! as
leaf. The result is a lazy list \lstinline!mc! of Monte Carlo
iterations, where an iteration consists of a Monte Carlo tree and
possibly a proof discovered during the simulation performed in the
iteration. The first \lstinline!maxIterations! are considered: When
\lstinline!maxIterations! is set to \(0\), proof search behaves like
leanCoP, and in case it is set to \(\infty\), the whole proof search is
performed in the MCPS part. As MCPS is performed lazily, MCPS is
performed for less than \lstinline!maxIterations! iterations when it
discovers some proof contributing to the final closed derivation. Here,
the lazy list characterisation introduced in \autoref{lazylist} turns
out to be permit a very concise implementation as well as an easy
integration of techniques such as restricted backtracking. As soon as
all proofs discovered during MCPS were considered, the tree \(T\) of the
final Monte Carlo iteration \lstinline!last mc! is obtained and the
children of the root of \(T\) are sorted by decreasing average
\(T\)-descendant reward \(\overline \rho _T\). Finally, the last applied
proof step of each child is processed like in the lazy list
implementation.

\begin{lstlisting}[language=Haskell, caption=Monte Carlo Proof Search as advisor., label=lst:montecop]
prove [] path lim sub = [sub]
prove (lit : cla) path lim sub =
  let
    mc = take maxIterations (mcps lit path sub)
    proofs1 = mapMaybe getProof mc
    proofs2 = last mc & root & children & sortOn avgReward & concatMap
      (\ child -> case lastStep child of
        Reduction sub1 -> [sub1]
    Extension (sub1, cla1) ->
          if lim <= 0 then []
          else prove cla1 (lit : path) (lim - 1) sub1)
  in concatMap (prove cla path lim) (proofs1 ++ proofs2)
\end{lstlisting}

The array substitution technique from \autoref{unification} requires
that the proof search always backtracks only to states whose
substitution is a subset of the current state's substitution. However,
because this requirement is not fulfilled for MCPS, we use association
lists for substitutions.

\subsection{Evaluation}\label{evaluation}

We evaluated the presented heuristics on the bushy MPTP2078 problems,
with definitional clausification and a timeout of 10s for each problem.
Before evaluation, we collected training data for the machine learning
heuristics by running leanCoP on all bushy problems with a timeout of
60s.

The base configuration of monteCoP uses the open branches probability
(see \autoref{childprob}), the branch ratio reward (see
\autoref{reward}), and the minimal expansion policy 1 (see
\autoref{exppol}), where the maximal simulation depth \(s_{\max} = 50\),
the exploration constant \(C_p = 1\), and the maximal number of MCTS
iterations \lstinline!maxIterations! \(= \infty\). For any heuristic
\(h\) not used in the base configuration, we replaced the default
heuristic with \(h\) and evaluated the resulting configuration. The
results are shown in \autoref{tab:montecarlo-results}. We can see that
the heuristics that most improve the base configuration are the
machine-learnt closability reward and the minimal expansion policy 2.

\begin{table}

\caption{Comparison of Monte Carlo heuristics. Iterations, simulation
steps and discrimination ratio are averages on the 196 problems solved
by all configurations. \label{tab:montecarlo-results}}

\begin{tabular}{lrrrr}

\toprule

Configuration & Iterations & Sim. steps & Discr. & Solved\tabularnewline

\midrule

Base & 116.46 & 1389.82 & 1.37 & 332\tabularnewline
Uniform probability & 949.62 & 17539.59 & 1.31 & 237\tabularnewline
NB probability & 528.39 & 8014.03 & 1.35 & 248\tabularnewline
Random reward & 104.88 & 1167.98 & 1.19 & 364\tabularnewline
Branch weight reward & 108.13 & 1268.88 & 1.12 & 334\tabularnewline
ML closability
reward & 108.52 & 1151.61 & \textbf{2.30} & \textbf{367}\tabularnewline
Default exp. pol. & 371.81 & 4793.58 & 1.38 & 328\tabularnewline
Minimal exp. pol. 2 & 224.72 & 2769.12 & 1.40 & 348\tabularnewline

\bottomrule

\end{tabular}

\end{table}

We explored a range of values for several numeric parameters, for which
we show results in \autoref{fig:param-discussion}: The maximal number of
MCTS iterations \lstinline!maxIterations! performs best between 20 and
40, see \autoref{fig:maxiters}: Below 20, MCTS can not provide any
meaningful quality estimates, and above 40, the quality estimates do not
significantly improve any more, while costing computational resources.
The exploration constant \(C_p \approx 0.75\) gives best results, where
the machine-learnt closability reward achieves a local optimum, see
\autoref{fig:exploration}: At such an optimum, exploration and
exploitation combine each other best, therefore the existence of such an
optimum is a sanity check for reward heuristics (which the branch ratio
reward does not pass). The maximal simulation depth
\(s_{\max} \approx 20\) seems to perform best, see
\autoref{fig:simdepth}. Above this value, the number of solved problems
decreases, since the number of actually performed simulation steps
decreases, as shown in \autoref{fig:simdepth-ss}. This might be
explained by the fact that at higher simulation depths, the
computational effort to calculate the set of possible steps increases,
for example because the substitution contains more and larger elements.

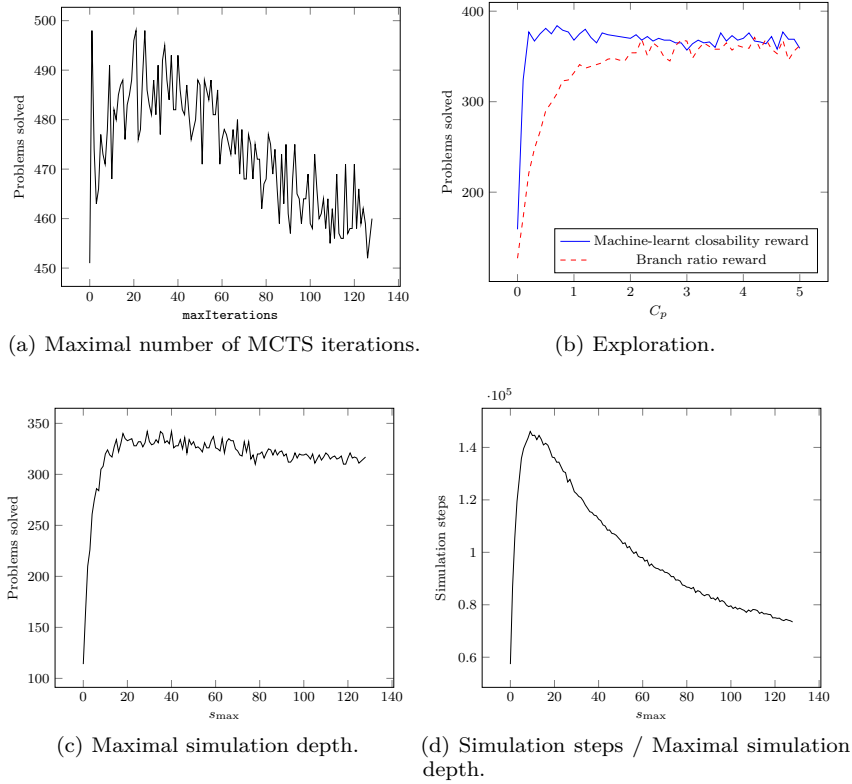
\begin{figure}[ht!]
  \subfloat[Maximal number of MCTS iterations.]{
    \begin{tikzpicture}[scale=0.65]
  \begin{axis}
  [ legend pos=south east
    , xlabel=\texttt{maxIterations}
  , ylabel=Problems solved
  ]
    \addplot [mark=none] table {tikz/montecop-170209-maxiters};
  \end{axis}
\end{tikzpicture}
    \label{fig:maxiters}
  }
  \subfloat[Exploration.]{
    \begin{tikzpicture}[scale=0.65]
  \begin{axis}
  [ legend pos=south east
  , xlabel=$C_p$
  , ylabel=Problems solved
  ]
    \addplot+ [mark=none] table {tikz/montecop-170209-mlrew-exploration};
    \addplot+ [dashed,mark=none] table {tikz/montecop-170209-exploration};
    \legend{Machine-learnt closability reward,Branch ratio reward}
  \end{axis}
\end{tikzpicture}
    \label{fig:exploration}
  }

  \subfloat[Maximal simulation depth.]{
    \begin{tikzpicture}[scale=0.65]
  \begin{axis}
  [ legend pos=south east
  , xlabel=$s_{\max}$
  , ylabel=Problems solved
  ]
    \addplot [mark=none] table {tikz/montecop-170209-simdepth};
  \end{axis}
\end{tikzpicture}
    \label{fig:simdepth}
  }
  \subfloat[Simulation steps / Maximal simulation depth.]{
    \begin{tikzpicture}[scale=0.65]
  \begin{axis}
  [ legend pos=south east
  , xlabel=$s_{\max}$
  , ylabel=Simulation steps
  ]
    \addplot [mark=none] table {tikz/montecop-170209-simdepth-simsteps};
  \end{axis}
\end{tikzpicture}
    \label{fig:simdepth-ss}
  }

  \caption{Parameter influence.}
  \label{fig:param-discussion}
\end{figure}

We adapted the base configuration to use the best heuristics from
\autoref{tab:montecarlo-results} and the best values for parameters
discussed in \autoref{fig:param-discussion}, yielding \(s_{\max} = 20\),
\(C_p = 0.75\), and \lstinline!maxIterations! \(= 27\). This improved
configuration solves 538 problems, compared to 509 solved by the best
single leanCoP strategy.

\section{Clausal and Nonclausal Proof
Certification}\label{certification}

In this section, we show how to certify connection tableaux proofs by
reconstructing them in HOL Light. To abstract from the technical
details, we give translations for both clausal and nonclausal versions
of connection proofs to LK \citep{gentzen1935-schliessen}. As the
implementation can be used for regular proof search in HOL Light, we
evaluate its performance on HOL Light problem sets.

\subsection{Converting HOL to FOL
Problems}\label{converting-hol-to-fol-problems}

To use leanCoP respectively nanoCoP as a proof tactic in HOL Light, it
is necessary to convert a given proof goal from HOL to FOL. For this, we
reuse a large part of the MESON \citep{harrison1996-meson}
infrastructure, such as instantiation of higher-order axioms. Once we
are left with a first-order problem
\((A_1 \land \dots \land A_n) \rightarrow C\), we transform it similarly
to \autoref{preprocessing}. From the resulting formula, we create a
matrix \(M\) and a literal database as explained in \autoref{matlitdb},
and search for a proof. The resulting proof consists of a connection
proof tree and a global substitution \(\sigma\). Given this information,
we want to construct a proof of \(M \vdash \bot\), which is written in
LK as \(M \vdash\). We show such a translation method both from clausal
and nonclausal proofs to LK in the next sections.

\subsection{Connection Calculi for Proof
Translation}\label{conn-calc-recon}

In the presentation of the connection calculi in
\citep{otten2011-nonclausal}, all proof rules have a fixed number of
premises. To ease the translation of proofs, we present slightly
reformulated versions of the calculi. We introduce the following
notation for rules with an arbitrary number of premises:
\[\parbox[c]{\hsize}
  {\begin{prfenv} \prftree{\displaystyle{\bigwedge} _ i P_i}{C} \end{prfenv}}
\quad \equiv \quad
\parbox[c]{\hsize}
  {\begin{prfenv} \prftree{P_1}{\dots}{P_n}{C} \end{prfenv}}\] The
reformulated calculi for translation are shown in
\autoref{fig:clausal-calculus-recon} and
\autoref{fig:nonclausal-calculus-recon}. The words of the original
calculi were \(\langle C, M, Path \rangle\). In the reformulated
calculi, the words are \(\langle X, M, Path \rangle\), where \(X\)
denotes an arbitrary clause element, i.e.~a matrix or a literal.
Furthermore, in the new calculi, the axiom rule becomes obsolete. Proofs
can be trivially translated between the connection calculi in this
chapter and those shown in \autoref{conncalculus}.

\begin{figure}
  \begin{tabular}{r m{6cm}}
  \parbox[c]{\hsize}{
    \begin{prfenv}
      \prftree[r]{Start}
      {\prfassumption{\displaystyle{\bigwedge}_i \langle X_i, M, \{\} \rangle}}
      {\varepsilon, M, \varepsilon}
    \end{prfenv}}
  & where $\{X_1, \dots, X_n\}$ is copy of $C \in M$
  \\
  \parbox[c]{\hsize}{
    \begin{prfenv}
      \prftree[r]{Reduction}{\phantom{M}}{L, M, Path \cup \{L'\}}
    \end{prfenv}}
  & where $\sigma(L) = \sigma(\overline{L'})$
  \\
  \parbox[c]{\hsize}{
    \begin{prfenv}
      \prftree[r]{Extension}
      {\prfassumption{\displaystyle{\bigwedge}_i \langle L_i, M, Path \cup \{L\} \rangle}}
      {L, M, Path}
    \end{prfenv}}
  & where $\{L_1, \dots, L_n\} \cup \{L'\}$ is copy of $C \in M$
    and $\sigma(L) = \sigma(\overline{L'})$
\end{tabular}

  \caption{Clausal connection calculus for translation.}
  \label{fig:clausal-calculus-recon}
\end{figure}

\begin{figure}
  \begin{tabular}{r m{5cm}}
  \parbox[c]{\hsize}{
    \begin{prfenv}
      \prftree[r]{Start}
      {\prfassumption{\displaystyle{\bigwedge}_i \langle X_i, M, \{\} \rangle}}
      {\varepsilon, M, \varepsilon}
    \end{prfenv}}
  & where $\{X_1, \dots, X_n\}$ is copy of $C \in M$
  \\
  \parbox[c]{\hsize}{
    \begin{prfenv}
      \prftree[r]{Reduction}{\phantom{M}}{L, M, Path \cup \{L'\}}
    \end{prfenv}}
  & where $\sigma(L) = \sigma(\overline{L'})$
  \\
  \parbox[c]{\hsize}{
    \begin{prfenv}
      \prftree[r]{Extension}
      {\prfassumption{\displaystyle{\bigwedge}_i \langle X_i, M[C_1 \backslash C_2], Path \cup \{L\} \rangle}}
      {L, M, Path}
    \end{prfenv}}
  & where $\{X_1, \dots, X_n\}$ is the $\beta$-clause of $C_2$ wrt. $L'$,
    $C_2$ is copy of $C_1$, $C_1$ is e-clause of $M$ wrt. $Path \cup \{L\}$,
    $C_2$ contains $L'$ with $\sigma(L) = \sigma(\overline{L'})$
  \\
  \parbox[c]{\hsize}{
    \begin{prfenv}
      \prftree[r]{Decomposition}
      {\prfassumption{\displaystyle{\bigwedge}_i \langle X_i, M, Path \rangle}}
      {M', M, Path}
    \end{prfenv}}
  & where $\{X_1, \dots, X_n\} \in M'$
\end{tabular}

  \caption{Nonclausal connection calculus for translation.}
  \label{fig:nonclausal-calculus-recon}
\end{figure}

\begin{example}\label{ex:conn-calc-recon}Consider \autoref{ex:conn-calc}
on page~\pageref{ex:conn-calc}. For the matrices \(M\) respectively
\(M'\), proofs in the connection calculi for translation are given in
\autoref{fig:nonclausal-calc-ex-recon} respectively
\autoref{fig:clausal-calc-ex-recon}.\end{example}

\begin{figure}
  $\input{prftree/nonclausal-calc-ex-recon2}$
  \caption{Proof in the nonclausal connection calculus for translation.}
  \label{fig:nonclausal-calc-ex-recon}
\end{figure}

\begin{figure}
  $\input{prftree/clausal-calc-ex-recon2}$
  \caption{Proof in the clausal connection calculus for translation.}
  \label{fig:clausal-calc-ex-recon}
\end{figure}

\subsection{Connection Proof
Translation}\label{connection-proof-translation}

We translate connection proof trees recursively by distinguishing the
different rules of the calculus. We denote by \([\Gamma \vdash]\) the LK
translation of the connection proof for \(\Gamma\). In the LK
translation, matrices and clauses represent formulas, where a matrix is
a conjunction of clauses and a clause is (potentially universally
quantified) disjunction. We write that \(C\) is in \(M\) iff
\(M = C_1 \land \dots \land C_n\) with \(C = C_i\) for some \(i\) with
\(1 \leq i \leq n\).

We use an LK rule \(\land\)L which extracts a conjunct from a
conjunction, however keeping the conjunction in the context:
\[\begin{prfenv}
  \prftree[r]{$\land$L}
    {\Gamma, C_i, C_1 \land \dots \land C_n \vdash \Delta}
    {\Gamma, C_1 \land \dots \land C_n \vdash \Delta}
\end{prfenv}\] Furthermore, we use an LK rule \(\bot\)L which derives
\(\bot\) from two complementary literals in the context:
\[\begin{prfenv}
  \prftree[r]{$\bot$L}
    {\Gamma, A, \lnot A \vdash}
\end{prfenv}\]

Let us start with the two rules that are translated the same way for
clausal and nonclausal proofs, namely the start and the reduction rule.
The translation of these rules is shown in \autoref{fig:common-transl}.

For the start rule, the translation obtains the formula corresponding to
the clause \(C\) with the \(\land\)L rule, and instantiates it with the
\(\forall\)L rule. The substitution \(\sigma\) is used to determine the
instantiations, where fresh names are invented when a variable is
unbound in the substitution. Then, the sequent is split into several
subsequents \([X_i, M \vdash]\), which represent the translations of the
connection proofs for \(\langle X_i, M, \{\} \rangle\).\footnote{In the
  clausal setting, \(X_i\) could be written as \(L_i\), but as the same
  rule is used in the nonclausal setting, where \(X_i\) can represent
  either a literal or a matrix, we write \(X_i\) for the common rules.}

\begin{figure}
  \begin{tabular}{c c}
  \toprule
  Connection Calculus & LK \\
  \midrule
  \makecell{
    \begin{prfenv}
      \prftree[r]{Start}
      {\prfassumption{X_1, M, \{\}}}
      {\prfassumption{\dots}}
      {\prfassumption{X_n, M, \{\}}}
      {\varepsilon, M, \varepsilon}
    \end{prfenv} \\
    where $\{X_1, \dots, X_n\}$ is a copy of $C \in M$}
  &
  \makecell{
    \begin{prfenv}
      \prftree[r]{$\land$L}
        {\prftree[r]{$\forall$L}
          {\prftree[r]{$\lor$L}
            {\prfboundedassumption{X_1, M, \{\} \vdash}}
            {\prfassumption{\dots}}
            {\prfboundedassumption{X_n, M, \{\} \vdash}}
            {X_1 \lor \dots \lor X_n, M \vdash}}
          {\forall \vec x. (X_1 \lor \dots \lor X_n), M \vdash}}
        {M \vdash}
    \end{prfenv} \\
    where $\forall \vec x. (X_1 \lor \dots \lor X_n)$ in $M$
  }
  \\
  \midrule
  \makecell{
    \begin{prfenv}
      \prftree[r]{Reduction}{\phantom{M}}{L, M, Path \cup \{L'\}}
    \end{prfenv} \\
    where $\sigma(L) = \sigma(\overline{L'})$
  }
  &
  \makecell{
    \begin{prfenv}
      \prftree[r]{$\bot$L}{\phantom{M}}{L, M, Path \cup \{\overline L\} \vdash}
    \end{prfenv}
  } \\
  \bottomrule
\end{tabular}

  \caption{LK translation of common connection calculus rules.}
  \label{fig:common-transl}
\end{figure}

\subsubsection{Clausal Proof Translation}\label{cl-transl}

The translation of the clausal extension rule (shown in
\autoref{fig:clausal-calculus-recon}) is given in
\autoref{fig:ext-c-transl}. First, \(L, M, Path \vdash\) is transformed
to the equivalent \(M, P \vdash\), where \(P = Path \cup \{L\}\).
Structurally, the remaining translation resembles that of the start
rule, with the exception that it additionally closes a proof branch
containing the negated literal \(\overline L\).

\begin{figure}
  \input{prftree/ext-c-transl}
  \caption{LK translation of the clausal extension rule.}
  \label{fig:ext-c-transl}
\end{figure}

\subsubsection{Nonclausal Proof
Translation}\label{nonclausal-proof-translation}

We now proceed with the translation of nonclausal connection proofs,
using the calculus introduced in \autoref{fig:nonclausal-calculus}. The
LK context in the translation of nonclausal proofs now has the shape
\(X, \vec M, Path\). \(X\) refers to either a literal or a matrix,
whereas \(X\) represented just literals in \autoref{cl-transl}.
Furthermore, we use a set of matrices \(\vec M\), instead of a single
matrix \(M\) as in the clausal case. During translation, \(\vec M\) is
extended such that for each word \(\langle L, M, Path \rangle\) in the
connection calculus and its corresponding LK sequent
\(L, \vec M, Path \vdash\), the e-clauses of \(M\) wrt
\(Path \cup \{L\}\) are the clauses \(C\) for which \(C\) in \(M'\) and
\(M' \in \vec M\). We will see this in detail in the explanation for the
extension rule.

The LK translation of nonclausal proofs reuses the translations of the
start and the reduction rules given in \autoref{fig:common-transl}.
However, occurrences of \(M\) in the LK translation are replaced by
\(\vec M\). The start rule uses \(\vec M = \{M\}\), i.e. \(\vec M\)
contains only the initial problem matrix \(M\).

The decomposition rule of the nonclausal calculus can be seen as a
generalisation of the start rule. We give its translation to LK in
\autoref{fig:decomposition-transl}.

\begin{figure}
  \begin{tabular}{c c}
  \toprule
  Connection Calculus & LK \\
  \midrule
  \makecell{
    \begin{prfenv}
      \prftree[r]{Decomposition}
      {\prfassumption{\displaystyle{\bigwedge}_i \langle X_i, M, Path \rangle}}
      {M', M, Path}
    \end{prfenv} \\
    where $\{X_1, \dots, X_n\} \in M'$}
  &
  \makecell{
    \begin{prfenv}
      \prftree[r]{$\land$L}
      {\prftree[r]{$\forall$L}
        {\prftree[r]{$\lor$L}
            {\prfboundedassumption{X_1, \vec M', Path \vdash}}
            {\prfassumption{\dots}}
            {\prfboundedassumption{X_n, \vec M', Path \vdash}}
          {X_1 \lor \dots \lor X_n, \vec M', Path \vdash}}
	    {\forall \vec x. (X_1 \lor \dots \lor X_n), \vec M', Path \vdash}}
      {M', \vec M, Path \vdash}
    \end{prfenv} \\
    where $\forall \vec x. (X_1 \lor \dots \lor X_n)$ in $M'$ \\
    and   $\vec M' = \{M'\} \cup \vec M$
  } \\
  \bottomrule
\end{tabular}

  \caption{LK translation of the decomposition rule.}
  \label{fig:decomposition-transl}
\end{figure}

Let us now consider a nonclausal extension step applied to
\(\langle L, M, Path \rangle\). Let \(C_1\) denote the e-clause of \(M\)
wrt \(Path \cup \{L\}\) that was used for the extension step. By
construction of \(\vec M\) mentioned above, \(C_1\) is some clause in
\(M_1 \in \vec M\). Furthermore, let \(\beta_ 1\) be the
\(\beta\)-clause of \(C_1\) wrt \(\overline L\). Then we can find some
\(m\) such that \(M_1\), \(C_1\) and \(\beta _1\) can be written as in
\autoref{fig:matrix-beta}.

\begin{figure}
  \input{prftree/matrix-beta}
  \caption{Definition of matrix $M_i$, clause $C_i$, and $\beta$-clause $\beta_i$.}
  \label{fig:matrix-beta}
\end{figure}

\begin{figure}
  \input{prftree/ext-nc-transl}
  \caption{LK translation of the nonclausal extension rule.}
  \label{fig:ext-nc-transl}
\end{figure}

The translation of the nonclausal extension rule is shown in
\autoref{fig:ext-nc-transl}. We first transform
\(L, \vec M, Path \vdash\) to the equivalent \(\vec M^0, P \vdash\), as
\(\vec M^0 = \vec M\) and \(P = Path \cup \{L\}\). We then determine
\(M_1 \in \vec M\) and put it into the context by contraction (CL).

Now we recursively prove the sequent \(M_i, M_{i-1}, P \vdash\): If
\(M_i\) is the literal \(\overline L\), we prove the sequent
\(\overline L, \vec M^m, P \vdash\) with the \(\bot\)L rule. Otherwise,
we proceed as follows: First, we choose the appropriate clause \(C_i\)
of \(M_i\) that corresponds to \(\beta _i\). In the same step, we merge
\(M_i\) with \(\vec M^{i-1}\), yielding \(\vec M^i\). After the
instantiation of \(C_i\), the clause elements \(X _ {i,1}\) to
\(X _ {i,n_i}\) give rise to several proof branches where all but one
are closed by translation of the proof branches of the connection proof.
The one remaining clause element \(M_{i+1}\) gives rise to a sequent
\(M_{i+1}, \vec M^i, P \vdash\), which we translate by recursion. This
concludes the translation of the extension rule.

\begin{figure}
  \input{prftree/nonclausal-lk-ex}
  \caption{Translation of a nonclausal proof to LK.}
  \label{fig:nonclausal-lk-ex}
\end{figure}

\begin{example}Consider the nonclausal proof given in
\autoref{fig:nonclausal-calc-ex-recon}. We show its translation to LK in
\autoref{fig:nonclausal-lk-ex}, where boxed sequents indicate words of
the original proof. We use \(F\) from \autoref{ex:conn-calc} to define
\[\begin{aligned}
  \vec M_0 &= \{F\}, \\
  \vec M_1 &= \vec M_0 \cup \{\lnot P(s^2 a) \land (P(sa) \lor \lnot Q)\}, \\
  \vec M_2 &= \vec M_1 \cup \{\lnot P(s^3 a) \land (P(s^2 a) \lor \lnot Q)\}.
  \end{aligned}\]\end{example}

\subsection{Evaluation}\label{evaluation-1}

We evaluate our HOL Light proof search tactics based on leanCoP and
nanoCoP and compare their performance with the Metis
\citep{faerber2015-metis} and MESON \citep{harrison1996-meson} tactics
integrated into HOL Light. Similarly to \citep{kaliszyk2015-holcop}, we
disable splitting for MESON. As evaluation datasets, we use toplevel and
MESON goals from core HOL Light as well as from Flyspeck, see
\autoref{tab:hol-eval-datasets}.\footnote{Note that this table mentions
  a larger number of MESON goals than \autoref{fof-datasets}. This is
  because we consider for this evaluation also those problems that are
  solved by the first-order export.} We use the
\href{https://github.com/jrh13/hol-light/commit/08f4461}{Git version
08f4461 of HOL Light from March 2017}, running every tactic with a
timeout of 10s on each problem.

The results are shown in \autoref{tab:hol-results}. As many problems in
HOL Light are solved using either MESON or Metis, the problems are
likely biased towards these two provers. Furthermore, the fact that both
MESON and Metis are clausal might have shaped the design of the theorems
in HOL Light towards solvability by clausal provers.

Comparing \autoref{tab:hol-results} with \autoref{tab:table1}, we notice
that the stand-alone connection provers perform better than their
counterparts integrated in HOL Light. Apart from different
preprocessing, this can be explained by different array access
performance: Array access is more than 30 times faster in native OCaml
programs, compared to programs compiled in OCaml's toplevel (as used in
HOL Light). This heavily affects our connection provers, as fast
unification via arrays is critical for their performance, see
\autoref{unification}.

\begin{table}

\caption{HOL Light evaluation datasets and number of contained problems.
\label{tab:hol-eval-datasets}}

\begin{tabular}{llll}

\toprule

HL-top & HL-meson & FS-top & FS-meson\tabularnewline

\midrule

2499 & 1119 & 27112 & 44468\tabularnewline

\bottomrule

\end{tabular}

\end{table}

\begin{table}

\caption{HOL Light tactic results. \label{tab:hol-results}}

\begin{tabular}{lrrrr}

\toprule

Prover & HL-top & HL-meson & FS-top & FS-meson\tabularnewline

\midrule

Metis & 807 & 1029 & 4626 & 42829\tabularnewline
MESON & 736 & 900 & 4221 & 39227\tabularnewline
\midrule leanCoP\(+\)cut & 724 & 948 & 3714 & 39922\tabularnewline
leanCoP\(-\)cut & 717 & 844 & 3800 & 38528\tabularnewline
nanoCoP\(+\)cut & 538 & 802 & 2743 & 34213\tabularnewline
nanoCoP\(-\)cut & 550 & 811 & 2351 & 34769\tabularnewline

\bottomrule

\end{tabular}

\end{table}

\section{Related Work}\label{related-work}

A number of related works has already been discussed in previous
sections. In particular in \autoref{conncalculus}, we introduced the
connection calculus \citep{bibel1991-connection} as a variant of
tableaux \citep{letz2001-connection}, we discussed its implementation in
the leanCoP theorem prover \citep{otten2003-leancop}, a number of
improvements introduced in the second version of leanCoP
\citep{otten2008-leancop} including restricted backtracing
\citep{otten2010-cut}, and the nonclausal variant of the connection
calculus \citep{otten2011-nonclausal} together with its implementation
\citep{otten2016-nanocop}.

The compact Prolog implementation of theorem provers following the
\emph{lean} architecture made it attractive for many experiments both
with the calculus and with the implementation. The intuitionistic
version of leanCoP \citep{otten2005-ileancop} became the state-of-art
prover for first-order problems in intuitionistic logic
\citep{raths2007-iltp}. Connections have also been considered for
first-order modal logic in mleanCoP \citep{otten2014-mleancop}, for
higher-order logic \citep{andrews1989-hoconn} and for linear logic
\citep{galmiche2000-linlogconn}. Various implementation modifications
can be performed very elegantly, such as search strategies, scheduling,
randomization of the order of proof search steps
\citep{raths2008-randocop}, and internal guidance
\citep{urban2011-malecop, kaliszyk2015-femalecop}.

A number of early learning and data based approaches to guide automated
theorem provers has been surveyed in \citep{denzinger1999-experience}.
The Prover9 hints method \citep{veroff1996-hints} allows the user to
specify (an often large set of) clauses to treat in a special way. A
similarly working \emph{watch list} has been later integrated in E,
along with other learning mechanisms \citep{schulz2001-learning}.
Further methods for guiding the actual proof search of ATPs using
machine learning have been considered in the integration of a Naive
Bayesian classifier to select next proof actions in Satallax
\citep{faerber2016-satallax}, as well as in Enigma
\citep{jakubuv2017-enigma} where the clause selection in E uses a
tree-based n-gram approach to approximate similarity to the learned
proofs using a support vector machine classifier. Holophrasm
\citep{whalen2016-holophrasm} introduces a theorem prover architecture
using GRU neural networks to guide the proof search of a tableaux style
proof process of MetaMath. TensorFlow neural network guidance was
integrated in E \citep{loos2017-deepnetwork}, showing that with batching
and hybrid heuristics, it can solve a number of problems other
strategies cannot solve. Finally, various reasons as to why the
connection calculus is well suited for machine learning techniques,
especially deep learning, are considered in \citep{bibel2017-deeplean}.

The main use of machine learning in automated and interactive theorem
provers today is to reduce original problems before the actual proof
search. Machine learning based methods
\citep{kuehlwein2013-mash, blanchette2016-mash} improve on and
complement the various ATP heuristics \citep{hoder2011-sine} and ITP
heuristics \citep{meng2009-mepo}. The problem of selecting the most
useful lemmas for the given proof, refered to as ``premise selection''
or ``relevance filtering'' \citep{alama2014-premsel} nowadays uses
syntactic similarity approaches, simple Naive Bayes and k-NN based
classifiers, regression and kernel based methods
\citep{kuehlwein2012-premsel}, as well as deep neural networks
\citep{irving2016-deepmath}. This has become especially important in the
``large theory bench'' division added to the CADE Automated Systems
Competition in 2008 \citep{sutcliffe2009-j4}, with systems such as
MaLARea \citep{urban2008-malarea} and ET \citep{kaliszyk2015-et}
achieving notable results.

Theorem proving can be seen as a game -- for instance, it has been
modelled as a two-player game in the framework of game-theoretical
semantics \citep{hintikka1982-games}. Monte Carlo Tree Search (MCTS)
\citep{browne2012-survey} has been found to produce state-of-the-art
players for several games, most notably for the two-player game Go
\citep{silver2016-alphago}, but also for single-player games such as
SameGame \citep{schadd2012-samegame}. It therefore seems reasonable to
apply MCTS to the game of theorem proving. One-step lookahead can help
Vampire proof search \citep{hoder2016-selection}, suggesting that MCTS,
whose simulation phase can be seen as multi-step lookahead, can
effectively guide proof search.

Certification of ATP found proofs has been especially important for the
integration of ATPs into interactive proof assistants. Such components
provide automation in the form of proof tactics or automated
justification for smaller steps. HOL Light includes the certified proof
producing model elimination prover MESON \citep{harrison1996-meson}. The
paramodulation-based prover Metis \citep{hurd2003-metis} was designed
with a small certified proof core to simplify its integration with
interactive theorem provers \citep{faerber2015-metis}. Coq includes a
proof certifying version of the intuitionistic first-order automated
theorem prover JProver \citep{schmitt2001-jprover} and Matita includes a
proof certifying version of an ordered paramodulation prover
\citep{asperti2007-horecon}. A translation of connection tableaux proofs
to expansion trees which can be used for proof certification was studied
in \citep{reis2015-expansion}. An alternative approach to proof
certification is the usage of verified automated theorem provers
\citep{ridge2005-verifol}.

\section{Conclusion and Future work}\label{conclusion-and-future-work}

We have given an overview of our experiments conducted with connection
provers. First, we presented translations to functional programming
languages, exploring possibilities to increase the speed of proof search
while keeping the implementation as simple as possible. We showed that
the number of solved problems can be increased by up to 58.8\%, on one
dataset beating even E in automatic mode. Next, we discussed machine
learning integration in leanCoP via context-sensitive clause ordering
and Monte Carlo Tree Search, showing that both these techniques can
increase the number of solved problems, despite fewer inferences being
performed. Finally, we showed how to translate clausal and nonclausal
connection proofs to LK, yielding a usable proof search tactic for HOL
Light.

The performed machine learning experiments are promising enough to
justify the enhancement of Monte Carlo Proof Search with stronger
heuristics, such as neural networks. While we applied Monte Carlo Tree
Search to theorem proving as a single-player game, it could also be used
to treat theorem proving as two-player game.

The combination of several tools that are small, simple and
comprehensible can be more effective than a large, monolithic tool.
While the resulting connection provers cannot yet outperform larger
systems like Vampire \citep{kovacs2013-vampire} and E
\citep{schulz2013-e}, we hope that the insight gained by experiments
performed in connection provers might be used in their complex
counterparts. Connection provers might be candidates for the core of
future automated reasoning tools and artificial intelligence
experiments.

\begin{acknowledgement}We thank the reviewers of CPP, LPAR, and CADE for
their valuable comments. This work has been supported by a doctoral
scholarship of the University of Innsbruck, the Austrian Science Fund
(FWF) grant P26201, the European Research Council (ERC) grants no.
649043 \emph{AI4REASON} and no. 714034 \emph{SMART}, the Czech project
AI\&Reasoning CZ.02.1.01/0.0/0.0/15\_003/0000466, and the European
Regional Development Fund.\end{acknowledgement}

\bibliography{literature}

\begin{thebibliography}{83}
\ifx \bisbn   \undefined \def \bisbn  #1{ISBN #1}\fi
\ifx \binits  \undefined \def \binits#1{#1} \fi
\ifx \bauthor  \undefined \def \bauthor#1{#1} \fi
\ifx \bjtitle  \undefined \def \bjtitle#1{\textit{#1}}\fi
\ifx \batitle  \undefined \def \batitle#1{#1} \fi
\ifx \bctitle  \undefined \def \bctitle#1{#1} \fi
\ifx \bvolume  \undefined \def \bvolume#1{#1}\fi
\ifx \byear  \undefined \def \byear#1{#1} \fi
\ifx \bissue  \undefined \def \bissue#1{#1} \fi
\ifx \bfpage  \undefined \def \bfpage#1{#1} \fi
\ifx \blpage  \undefined \def \blpage #1{#1} \fi
\ifx \burl  \undefined \def \burl#1{#1} \fi
\ifx \doiurl  \undefined \def \doiurl#1{#1} \fi
\ifx \betal  \undefined \def \betal{et al.} \fi
\ifx \binstitute  \undefined \def \binstitute#1{#1} \fi
\ifx \beditor  \undefined \def \beditor#1{#1} \fi
\ifx \bpublisher  \undefined \def \bpublisher#1{#1} \fi
\ifx \bbtitle  \undefined \def \bbtitle#1{\textit{#1}} \fi
\ifx \bedition  \undefined \def \bedition#1{#1} \fi
\ifx \bseriesno  \undefined \def \bseriesno#1{#1} \fi
\ifx \blocation  \undefined \def \blocation#1{#1} \fi
\ifx \bsertitle  \undefined \def \bsertitle#1{\textit{#1}} \fi
\ifx \bsnm \undefined \def \bsnm#1{#1} \fi
\ifx \bsuffix \undefined \def \bsuffix#1{#1} \fi
\ifx \bparticle \undefined \def \bparticle#1{#1} \fi
\ifx \barticle \undefined \def \barticle#1{#1} \fi
\ifx \botherref \undefined \def \botherref #1{#1} \fi
\ifx \url \undefined \def \url#1{#1} \fi
\ifx \bchapter \undefined \def \bchapter#1{#1} \fi
\ifx \bbook \undefined \def \bbook#1{#1} \fi
\ifx \bcomment \undefined \def \bcomment#1{#1} \fi
\ifx \oauthor \undefined \def \oauthor#1{#1} \fi
\ifx \citeauthoryear \undefined \def \citeauthoryear#1{#1} \fi
\ifx \texttildelow  \undefined \def \texttildelow{\symbol{126}} \fi
\def \endbibitem {}
\ifx \bconflocation  \undefined \def \bconflocation#1{#1} \fi

\bibitem[\protect\citeauthoryear{Alama et~al.}{2014}]{alama2014-premsel}
\begin{barticle}
\bauthor{\bsnm{Alama}, \binits{Jesse}}, \bauthor{\binits{Tom} \bsnm{Heskes}},
  \bauthor{\binits{Daniel} \bsnm{K{\"{u}}hlwein}}, \bauthor{\binits{Evgeni}
  \bsnm{Tsivtsivadze}}, and \bauthor{\binits{Josef} \bsnm{Urban}}.
\byear{2014}.
\batitle{Premise selection for mathematics by corpus analysis and kernel
  methods}.
\bjtitle{J. Autom. Reasoning}
\bvolume{52} (\bissue{2}): \bfpage{191}--\blpage{213}.
doi:\doiurl{10.1007/s10817-013-9286-5}.
\end{barticle}
\endbibitem

\bibitem[\protect\citeauthoryear{Andrews}{1989}]{andrews1989-hoconn}
\begin{barticle}
\bauthor{\bsnm{Andrews}, \binits{Peter~B.}}
\byear{1989}.
\batitle{On connections and higher-order logic}.
\bjtitle{J. Autom. Reasoning}
\bvolume{5} (\bissue{3}): \bfpage{257}--\blpage{291}.
doi:\doiurl{10.1007/BF00248320}.
\end{barticle}
\endbibitem

\bibitem[\protect\citeauthoryear{Armando et~al.}{2008}]{DBLP:conf/cade/2008}
\begin{bbook}
\beditor{\bsnm{Armando}, \binits{Alessandro}}, \beditor{\binits{Peter}
  \bsnm{Baumgartner}}, and \beditor{\binits{Gilles} \bsnm{Dowek}}, eds.
\byear{2008}.
\bctitle{{IJCAR}}.
Vol. \bseriesno{5195}
of \bsertitle{{LNCS}}.
\bpublisher{Springer}.
\bisbn{978-3-540-71069-1}.
\end{bbook}
\endbibitem

\bibitem[\protect\citeauthoryear{Asperti and Tassi}{2007}]{asperti2007-horecon}
\begin{bchapter}
\bauthor{\bsnm{Asperti}, \binits{Andrea}}, and \bauthor{\binits{Enrico}
  \bsnm{Tassi}}.
\byear{2007}.
\bctitle{Higher order proof reconstruction from paramodulation-based
  refutations: The unit equality case}.
In \bbtitle{{MKM}},
eds. \beditor{\binits{Manuel} \bsnm{Kauers}}, \beditor{\binits{Manfred}
  \bsnm{Kerber}}, \beditor{\binits{Robert} \bsnm{Miner}}, and
  \beditor{\binits{Wolfgang} \bsnm{Windsteiger}}.
Vol. \bseriesno{4573}
of \bsertitle{{LNCS}},
\bfpage{146}--\blpage{160}.
\bpublisher{Springer}.
doi:\doiurl{10.1007/978-3-540-73086-6_14}.
\bisbn{978-3-540-73083-5}.
\end{bchapter}
\endbibitem

\bibitem[\protect\citeauthoryear{Beckert et~al.}{1993}]{beckert1993-delta}
\begin{bchapter}
\bauthor{\bsnm{Beckert}, \binits{Bernhard}}, \bauthor{\binits{Reiner}
  \bsnm{H{\"{a}}hnle}}, and \bauthor{\binits{Peter~H.} \bsnm{Schmitt}}.
\byear{1993}.
\bctitle{The even more liberalized $\delta$-rule in free variable semantic
  tableaux}.
In \bbtitle{Computational logic and proof theory, third {Kurt G{\"{o}}del}
  colloquium, {KGC}'93, {Brno}, {Czech Republic}, {August} 24-27, 1993,
  proceedings},
eds. \beditor{\binits{Georg} \bsnm{Gottlob}}, \beditor{\binits{Alexander}
  \bsnm{Leitsch}}, and \beditor{\binits{Daniele} \bsnm{Mundici}}.
Vol. \bseriesno{713}
of \bsertitle{{LNCS}},
\bfpage{108}--\blpage{119}.
\bpublisher{Springer}.
doi:\doiurl{10.1007/BFb0022559}.
\bisbn{3-540-57184-1}.
\end{bchapter}
\endbibitem

\bibitem[\protect\citeauthoryear{Berghofer et~al.}{2009}]{DBLP:conf/tphol/2009}
\begin{bbook}
\beditor{\bsnm{Berghofer}, \binits{Stefan}}, \beditor{\binits{Tobias}
  \bsnm{Nipkow}}, \beditor{\binits{Christian} \bsnm{Urban}}, and
  \beditor{\binits{Makarius} \bsnm{Wenzel}}, eds.
\byear{2009}.
\bctitle{{TPHOLs}}.
Vol. \bseriesno{5674}
of \bsertitle{{LNCS}}.
\bpublisher{Springer}.
doi:\doiurl{10.1007/978-3-642-03359-9}.
\bisbn{978-3-642-03358-2}.
\end{bbook}
\endbibitem

\bibitem[\protect\citeauthoryear{Bertot}{2008}]{bertot2008-coq}
\begin{bchapter}
\bauthor{\bsnm{Bertot}, \binits{Yves}}.
\byear{2008}.
\bctitle{A short presentation of {Coq}}.
In \bbtitle{{TPHOLs}},
eds. \beditor{\binits{Otmane~A{\"{\i}}t} \bsnm{Mohamed}},
  \beditor{\binits{C{\'{e}}sar~A.} \bsnm{Mu{\~{n}}oz}}, and
  \beditor{\binits{Sofi{\`{e}}ne} \bsnm{Tahar}}.
Vol. \bseriesno{5170}
of \bsertitle{{LNCS}},
\bfpage{12}--\blpage{16}.
\bpublisher{Springer}.
doi:\doiurl{10.1007/978-3-540-71067-7_3}.
\bisbn{978-3-540-71065-3}.
\end{bchapter}
\endbibitem

\bibitem[\protect\citeauthoryear{Bibel}{1983}]{bibel1983-matings}
\begin{barticle}
\bauthor{\bsnm{Bibel}, \binits{Wolfgang}}.
\byear{1983}.
\batitle{Matings in matrices}.
\bjtitle{Commun. {ACM}}
\bvolume{26} (\bissue{11}): \bfpage{844}--\blpage{852}.
doi:\doiurl{10.1145/182.183}.
\end{barticle}
\endbibitem

\bibitem[\protect\citeauthoryear{Bibel}{1991}]{bibel1991-connection}
\begin{bchapter}
\bauthor{\bsnm{Bibel}, \binits{Wolfgang}}.
\byear{1991}.
\bctitle{Perspectives on automated deduction}.
In \bbtitle{Automated reasoning: Essays in honor of {Woody Bledsoe}},
ed. \beditor{\binits{Robert~S.} \bsnm{Boyer}}.
\bsertitle{Automated reasoning series},
\bfpage{77}--\blpage{104}.
\bpublisher{Kluwer Academic Publishers}.
\bisbn{0-7923-1409-3}.
\end{bchapter}
\endbibitem

\bibitem[\protect\citeauthoryear{Bibel}{2017}]{bibel2017-deeplean}
\begin{bchapter}
\bauthor{\bsnm{Bibel}, \binits{Wolfgang}}.
\byear{2017}.
\bctitle{A vision for automated deduction rooted in the connection method}.
In \bbtitle{{TABLEAUX}},
eds. \beditor{\binits{Renate~A.} \bsnm{Schmidt}} and
  \beditor{\binits{Cl{\'{a}}udia} \bsnm{Nalon}}.
Vol. \bseriesno{10501}
of \bsertitle{{LNCS}},
\bfpage{3}--\blpage{21}.
\bpublisher{Springer}.
doi:\doiurl{10.1007/978-3-319-66902-1_1}.
\bisbn{978-3-319-66901-4}.
\end{bchapter}
\endbibitem

\bibitem[\protect\citeauthoryear{Biere et~al.}{2014}]{biere2014-satvampire}
\begin{bchapter}
\bauthor{\bsnm{Biere}, \binits{Armin}}, \bauthor{\binits{Ioan} \bsnm{Dragan}},
  \bauthor{\binits{Laura} \bsnm{Kov{\'{a}}cs}}, and \bauthor{\binits{Andrei}
  \bsnm{Voronkov}}.
\byear{2014}.
\bctitle{Experimenting with {SAT} solvers in {V}ampire}.
In \bbtitle{{MICAI} 2014. part {I}},
eds. \beditor{\binits{Alexander~F.} \bsnm{Gelbukh}},
  \beditor{\binits{F{\'{e}}lix Castro~{q}.} \bsnm{Espinoza}}, and
  \beditor{\binits{Sof{\'{\i}}a N. Galicia~{q}.} \bsnm{Haro}}.
Vol. \bseriesno{8856}
of \bsertitle{{LNCS}},
\bfpage{431}--\blpage{442}.
\bpublisher{Springer}.
doi:\doiurl{10.1007/978-3-319-13647-9_39}.
\bisbn{978-3-319-13646-2}.
\end{bchapter}
\endbibitem

\bibitem[\protect\citeauthoryear{Blanchette et~al.}{2016a}]{blanchette2016-qed}
\begin{barticle}
\bauthor{\bsnm{Blanchette}, \binits{Jasmin~Christian}},
  \bauthor{\binits{Cezary} \bsnm{Kaliszyk}}, \bauthor{\binits{Lawrence~C.}
  \bsnm{Paulson}}, and \bauthor{\binits{Josef} \bsnm{Urban}}.
\byear{2016}a.
\batitle{Hammering towards {QED}}.
\bjtitle{J. Formalized Reasoning}
\bvolume{9} (\bissue{1}): \bfpage{101}--\blpage{148}.
doi:\doiurl{10.6092/issn.1972-5787/4593}.
\end{barticle}
\endbibitem

\bibitem[\protect\citeauthoryear{Blanchette
  et~al.}{2016b}]{blanchette2016-mash}
\begin{barticle}
\bauthor{\bsnm{Blanchette}, \binits{Jasmin~Christian}}, \bauthor{\binits{David}
  \bsnm{Greenaway}}, \bauthor{\binits{Cezary} \bsnm{Kaliszyk}},
  \bauthor{\binits{Daniel} \bsnm{K{\"{u}}hlwein}}, and \bauthor{\binits{Josef}
  \bsnm{Urban}}.
\byear{2016}b.
\batitle{A learning-based fact selector for {Isabelle/HOL}}.
\bjtitle{J. Autom. Reasoning}
\bvolume{57} (\bissue{3}): \bfpage{219}--\blpage{244}.
doi:\doiurl{10.1007/s10817-016-9362-8}.
\end{barticle}
\endbibitem

\bibitem[\protect\citeauthoryear{Bove et~al.}{2009}]{bove2009-agda}
\begin{bchapter}
\bauthor{\bsnm{Bove}, \binits{Ana}}, \bauthor{\binits{Peter} \bsnm{Dybjer}},
  and \bauthor{\binits{Ulf} \bsnm{Norell}}.
\byear{2009}.
\bctitle{A brief overview of {Agda} - {A} functional language with dependent
  types}.
In \bbtitle{{TPHOLs}},
eds. \beditor{\binits{Stefan} \bsnm{Berghofer}}, \beditor{\binits{Tobias}
  \bsnm{Nipkow}}, \beditor{\binits{Christian} \bsnm{Urban}}, and
  \beditor{\binits{Makarius} \bsnm{Wenzel}}.
Vol. \bseriesno{5674}
of \bsertitle{{LNCS}},
\bfpage{73}--\blpage{78}.
\bpublisher{Springer}.
doi:\doiurl{10.1007/978-3-642-03359-9_6}.
\bisbn{978-3-642-03358-2}.
\end{bchapter}
\endbibitem

\bibitem[\protect\citeauthoryear{Browne et~al.}{2012}]{browne2012-survey}
\begin{barticle}
\bauthor{\bsnm{Browne}, \binits{Cameron}}, \bauthor{\binits{Edward~Jack}
  \bsnm{Powley}}, \bauthor{\binits{Daniel} \bsnm{Whitehouse}},
  \bauthor{\binits{Simon~M.} \bsnm{Lucas}}, \bauthor{\binits{Peter~I.}
  \bsnm{Cowling}}, \bauthor{\binits{Philipp} \bsnm{Rohlfshagen}},
  \bauthor{\binits{Stephen} \bsnm{Tavener}}, \bauthor{\binits{Diego~Perez}
  \bsnm{Liebana}}, \bauthor{\binits{Spyridon} \bsnm{Samothrakis}}, and
  \bauthor{\binits{Simon} \bsnm{Colton}}.
\byear{2012}.
\batitle{A survey of {Monte Carlo} tree search methods}.
\bjtitle{{IEEE} Trans. Comput. Intellig. and {AI} in Games}
\bvolume{4} (\bissue{1}): \bfpage{1}--\blpage{43}.
doi:\doiurl{10.1109/TCIAIG.2012.2186810}.
\end{barticle}
\endbibitem

\bibitem[\protect\citeauthoryear{Br{\"{u}}nnler and
  Metcalfe}{2011}]{DBLP:conf/tableaux/2011}
\begin{bbook}
\beditor{\bsnm{Br{\"{u}}nnler}, \binits{Kai}}, and \beditor{\binits{George}
  \bsnm{Metcalfe}}, eds.
\byear{2011}.
\bctitle{{TABLEAUX}}.
Vol. \bseriesno{6793}
of \bsertitle{{LNCS}}.
\bpublisher{Springer}.
doi:\doiurl{10.1007/978-3-642-22119-4}.
\bisbn{978-3-642-22118-7}.
\end{bbook}
\endbibitem

\bibitem[\protect\citeauthoryear{Carlson et~al.}{1999}]{carlson1999-snow}
\begin{botherref}
\oauthor{\bsnm{Carlson}, \binits{A.}}, \oauthor{\binits{C.} \bsnm{Cumby}},
  \oauthor{\binits{J.} \bsnm{Rosen}}, and \oauthor{\binits{D.} \bsnm{Roth}}.
1999.
{SNoW} user guide,
Technical Report UIUCDCS-R-99-2101,
University of Illinois at Urbana-Champaign.
\url{http://cogcomp.org/papers/CCRR99.pdf}.
\end{botherref}
\endbibitem

\bibitem[\protect\citeauthoryear{Denzinger
  et~al.}{1999}]{denzinger1999-experience}
\begin{botherref}
\oauthor{\bsnm{Denzinger}, \binits{Jörg}}, \oauthor{\binits{Matthias}
  \bsnm{Fuchs}}, \oauthor{\binits{Christoph} \bsnm{Goller}}, and
  \oauthor{\binits{Stephan} \bsnm{Schulz}}.
1999.
{Learning from Previous Proof Experience},
Technical Report AR99-4,
Institut f{\"u}r Informatik, Technische Universit{\"a}t M{\"u}nchen.
\end{botherref}
\endbibitem

\bibitem[\protect\citeauthoryear{F{\"{a}}rber and
  Brown}{2016}]{faerber2016-satallax}
\begin{bchapter}
\bauthor{\bsnm{F{\"{a}}rber}, \binits{Michael}}, and \bauthor{\binits{Chad~E.}
  \bsnm{Brown}}.
\byear{2016}.
\bctitle{Internal guidance for {Satallax}}.
In \bbtitle{{IJCAR}},
eds. \beditor{\binits{Nicola} \bsnm{Olivetti}} and \beditor{\binits{Ashish}
  \bsnm{Tiwari}}.
Vol. \bseriesno{9706}
of \bsertitle{{LNCS}},
\bfpage{349}--\blpage{361}.
\bpublisher{Springer}.
doi:\doiurl{10.1007/978-3-319-40229-1_24}.
\bisbn{978-3-319-40228-4}.
\end{bchapter}
\endbibitem

\bibitem[\protect\citeauthoryear{F{\"{a}}rber and
  Kaliszyk}{2015}]{faerber2015-metis}
\begin{bchapter}
\bauthor{\bsnm{F{\"{a}}rber}, \binits{Michael}}, and \bauthor{\binits{Cezary}
  \bsnm{Kaliszyk}}.
\byear{2015}.
\bctitle{Metis-based paramodulation tactic for {HOL Light}}.
In \bbtitle{{GCAI}},
eds. \beditor{\binits{Georg} \bsnm{Gottlob}}, \beditor{\binits{Geoff}
  \bsnm{Sutcliffe}}, and \beditor{\binits{Andrei} \bsnm{Voronkov}}.
Vol. \bseriesno{36}
of \bsertitle{Epic series in computing},
\bfpage{127}--\blpage{136}.
\bpublisher{EasyChair}.
\burl{http://www.easychair.org/publications/paper/Metis-based_Paramodulation_Tactic_for_HOL_Light}.
\end{bchapter}
\endbibitem

\bibitem[\protect\citeauthoryear{F{\"{a}}rber
  et~al.}{2017}]{faerber2017-montecop}
\begin{bchapter}
\bauthor{\bsnm{F{\"{a}}rber}, \binits{Michael}}, \bauthor{\binits{Cezary}
  \bsnm{Kaliszyk}}, and \bauthor{\binits{Josef} \bsnm{Urban}}.
\byear{2017}.
\bctitle{{Monte Carlo} tableau proof search}.
In \bbtitle{{CADE-26}},
ed. \beditor{\binits{Leonardo} \bparticle{de~}\bsnm{Moura}}.
Vol. \bseriesno{10395}
of \bsertitle{{LNCS}},
\bfpage{563}--\blpage{579}.
\bpublisher{Springer}.
doi:\doiurl{10.1007/978-3-319-63046-5_34}.
\bisbn{978-3-319-63045-8}.
\end{bchapter}
\endbibitem

\bibitem[\protect\citeauthoryear{Galmiche}{2000}]{galmiche2000-linlogconn}
\begin{barticle}
\bauthor{\bsnm{Galmiche}, \binits{Didier}}.
\byear{2000}.
\batitle{Connection methods in linear logic and proof nets construction}.
\bjtitle{Theor. Comput. Sci.}
\bvolume{232} (\bissue{1-2}): \bfpage{231}--\blpage{272}.
doi:\doiurl{10.1016/S0304-3975(99)00176-0}.
\end{barticle}
\endbibitem

\bibitem[\protect\citeauthoryear{Gentzen}{1935}]{gentzen1935-schliessen}
\begin{barticle}
\bauthor{\bsnm{Gentzen}, \binits{Gerhard}}.
\byear{1935}.
\batitle{Untersuchungen {\"u}ber das logische {Schlie{\ss}en}. {I}}.
\bjtitle{Mathematische Zeitschrift}
\bvolume{39} (\bissue{1}): \bfpage{176}--\blpage{210}.
doi:\doiurl{10.1007/BF01201353}.
\end{barticle}
\endbibitem

\bibitem[\protect\citeauthoryear{Giese and Ahrendt}{1999}]{giese1999-epsilon}
\begin{bchapter}
\bauthor{\bsnm{Giese}, \binits{Martin}}, and \bauthor{\binits{Wolfgang}
  \bsnm{Ahrendt}}.
\byear{1999}.
\bctitle{Hilbert's epsilon-terms in automated theorem proving}.
In \bbtitle{{TABLEAUX}},
ed. \beditor{\binits{Neil~V.} \bsnm{Murray}}.
Vol. \bseriesno{1617}
of \bsertitle{{LNCS}},
\bfpage{171}--\blpage{185}.
\bpublisher{Springer}.
doi:\doiurl{10.1007/3-540-48754-9_17}.
\bisbn{3-540-66086-0}.
\end{bchapter}
\endbibitem

\bibitem[\protect\citeauthoryear{Greenbaum}{1986}]{greenbaum1986-phd}
\begin{botherref}
\oauthor{\bsnm{Greenbaum}, \binits{Steven}}.
1986.
Input transformations and resolution implementation techniques for
  theorem-proving in first-order logic.
PhD diss,
University of Illinois at Urbana-Champaign.
\end{botherref}
\endbibitem

\bibitem[\protect\citeauthoryear{H{\"{a}}hnle}{2001}]{haehnle2001-tableaux}
\begin{bchapter}
\bauthor{\bsnm{H{\"{a}}hnle}, \binits{Reiner}}.
\byear{2001}.
\bctitle{Tableaux and related methods}.
In \bbtitle{Handbook of automated reasoning (in 2 volumes)},
eds. \beditor{\binits{John~Alan} \bsnm{Robinson}} and \beditor{\binits{Andrei}
  \bsnm{Voronkov}},
\bfpage{100}--\blpage{178}.
\bpublisher{Elsevier and {MIT} Press}.
\bisbn{0-444-50813-9}.
\end{bchapter}
\endbibitem

\bibitem[\protect\citeauthoryear{Hales et~al.}{2017}]{hales2017-kepler}
\begin{botherref}
\oauthor{\bsnm{Hales}, \binits{Thomas~C.}}, \oauthor{\binits{Mark}
  \bsnm{Adams}}, \oauthor{\binits{Gertrud} \bsnm{Bauer}},
  \oauthor{\binits{Dat~Tat} \bsnm{Dang}}, \oauthor{\binits{John}
  \bsnm{Harrison}}, \oauthor{\binits{Truong~Le} \bsnm{Hoang}},
  \oauthor{\binits{Cezary} \bsnm{Kaliszyk}}, \oauthor{\binits{Victor}
  \bsnm{Magron}}, \oauthor{\binits{Sean} \bsnm{McLaughlin}},
  \oauthor{\binits{Thang~Tat} \bsnm{Nguyen}}, \oauthor{\binits{Truong~Quang}
  \bsnm{Nguyen}}, \oauthor{\binits{Tobias} \bsnm{Nipkow}},
  \oauthor{\binits{Steven} \bsnm{Obua}}, \oauthor{\binits{Joseph}
  \bsnm{Pleso}}, \oauthor{\binits{Jason} \bsnm{Rute}}, \oauthor{\binits{Alexey}
  \bsnm{Solovyev}}, \oauthor{\binits{An~Hoai~Thi} \bsnm{Ta}},
  \oauthor{\binits{Trung~Nam} \bsnm{Tran}}, \oauthor{\binits{Diep~Thi}
  \bsnm{Trieu}}, \oauthor{\binits{Josef} \bsnm{Urban}},
  \oauthor{\binits{Ky~Khac} \bsnm{Vu}}, and \oauthor{\binits{Roland}
  \bsnm{Zumkeller}}.
2017.
A formal proof of the {Kepler} conjecture.
\textit{Forum of Mathematics, Pi}
5.
doi:\doiurl{10.1017/fmp.2017.1}.
\end{botherref}
\endbibitem

\bibitem[\protect\citeauthoryear{Harrison}{1996}]{harrison1996-meson}
\begin{bchapter}
\bauthor{\bsnm{Harrison}, \binits{John}}.
\byear{1996}.
\bctitle{Optimizing proof search in model elimination}.
In \bbtitle{{CADE-13}},
eds. \beditor{\binits{Michael~A.} \bsnm{McRobbie}} and
  \beditor{\binits{John~K.} \bsnm{Slaney}}.
Vol. \bseriesno{1104}
of \bsertitle{{LNCS}},
\bfpage{313}--\blpage{327}.
\bpublisher{Springer}.
doi:\doiurl{10.1007/3-540-61511-3_97}.
\bisbn{3-540-61511-3}.
\end{bchapter}
\endbibitem

\bibitem[\protect\citeauthoryear{Harrison}{2009}]{harrison2009-hollight}
\begin{bchapter}
\bauthor{\bsnm{Harrison}, \binits{John}}.
\byear{2009}.
\bctitle{{HOL Light}: An overview}.
In \bbtitle{{TPHOLs}},
eds. \beditor{\binits{Stefan} \bsnm{Berghofer}}, \beditor{\binits{Tobias}
  \bsnm{Nipkow}}, \beditor{\binits{Christian} \bsnm{Urban}}, and
  \beditor{\binits{Makarius} \bsnm{Wenzel}}.
Vol. \bseriesno{5674}
of \bsertitle{{LNCS}},
\bfpage{60}--\blpage{66}.
\bpublisher{Springer}.
doi:\doiurl{10.1007/978-3-642-03359-9_4}.
\bisbn{978-3-642-03358-2}.
\end{bchapter}
\endbibitem

\bibitem[\protect\citeauthoryear{Hilbert and
  Bernays}{1939}]{hilbert1939-grundlagen}
\begin{bbook}
\bauthor{\bsnm{Hilbert}, \binits{David}}, and \bauthor{\binits{Paul}
  \bsnm{Bernays}}.
\byear{1939}.
\bbtitle{{Grundlagen der Mathematik. II}}.
Vol. \bseriesno{50}
of \bsertitle{{Die Grundlehren der mathematischen Wissenschaften}}.
\bpublisher{Springer}.
\end{bbook}
\endbibitem

\bibitem[\protect\citeauthoryear{Hintikka}{1982}]{hintikka1982-games}
\begin{barticle}
\bauthor{\bsnm{Hintikka}, \binits{Jaakko}}.
\byear{1982}.
\batitle{Game-theoretical semantics: insights and prospects}.
\bjtitle{Notre Dame Journal of Formal Logic}
\bvolume{23} (\bissue{2}): \bfpage{219}--\blpage{241}.
doi:\doiurl{10.1305/ndjfl/1093883627}.
\end{barticle}
\endbibitem

\bibitem[\protect\citeauthoryear{Hoder and Voronkov}{2011}]{hoder2011-sine}
\begin{bchapter}
\bauthor{\bsnm{Hoder}, \binits{Krystof}}, and \bauthor{\binits{Andrei}
  \bsnm{Voronkov}}.
\byear{2011}.
\bctitle{Sine qua non for large theory reasoning}.
In \bbtitle{{CADE-23}},
eds. \beditor{\binits{Nikolaj} \bsnm{Bj{\o}rner}} and \beditor{\binits{Viorica
  Sofronie~{q}.} \bsnm{Stokkermans}}.
Vol. \bseriesno{6803}
of \bsertitle{{LNCS}},
\bfpage{299}--\blpage{314}.
\bpublisher{Springer}.
doi:\doiurl{10.1007/978-3-642-22438-6_23}.
\bisbn{978-3-642-22437-9}.
\end{bchapter}
\endbibitem

\bibitem[\protect\citeauthoryear{Hoder et~al.}{2016}]{hoder2016-selection}
\begin{bchapter}
\bauthor{\bsnm{Hoder}, \binits{Krystof}}, \bauthor{\binits{Giles}
  \bsnm{Reger}}, \bauthor{\binits{Martin} \bsnm{Suda}}, and
  \bauthor{\binits{Andrei} \bsnm{Voronkov}}.
\byear{2016}.
\bctitle{Selecting the selection}.
In \bbtitle{{IJCAR}},
eds. \beditor{\binits{Nicola} \bsnm{Olivetti}} and \beditor{\binits{Ashish}
  \bsnm{Tiwari}}.
Vol. \bseriesno{9706}
of \bsertitle{{LNCS}},
\bfpage{313}--\blpage{329}.
\bpublisher{Springer}.
doi:\doiurl{10.1007/978-3-319-40229-1_22}.
\bisbn{978-3-319-40228-4}.
\end{bchapter}
\endbibitem

\bibitem[\protect\citeauthoryear{Hurd}{2003}]{hurd2003-metis}
\begin{bchapter}
\bauthor{\bsnm{Hurd}, \binits{Joe}}.
\byear{2003}.
\bctitle{First-order proof tactics in higher-order logic theorem provers}.
In \bbtitle{Design and application of strategies/tactics in higher order logics
  ({STRATA} 2003)},
eds. \beditor{\binits{Myla} \bsnm{Archer}}, \beditor{\binits{Ben~Di}
  \bsnm{Vito}}, and \beditor{\binits{C{\'{e}}sar} \bsnm{Mu{\~{n}}oz}}.
\bsertitle{{NASA} technical reports},
\bfpage{56}--\blpage{68}.
\burl{http://www.gilith.com/research/papers}.
\end{bchapter}
\endbibitem

\bibitem[\protect\citeauthoryear{Irving et~al.}{2016}]{irving2016-deepmath}
\begin{bchapter}
\bauthor{\bsnm{Irving}, \binits{Geoffrey}}, \bauthor{\binits{Christian}
  \bsnm{Szegedy}}, \bauthor{\binits{Alexander~A.} \bsnm{Alemi}},
  \bauthor{\binits{Niklas} \bsnm{E{\'{e}}n}}, \bauthor{\binits{Fran{\c{c}}ois}
  \bsnm{Chollet}}, and \bauthor{\binits{Josef} \bsnm{Urban}}.
\byear{2016}.
\bctitle{{DeepMath} - deep sequence models for premise selection}.
In \bbtitle{{NIPS}},
eds. \beditor{\binits{Daniel~D.} \bsnm{Lee}}, \beditor{\binits{Masashi}
  \bsnm{Sugiyama}}, \beditor{\binits{Ulrike} \bparticle{von }\bsnm{Luxburg}},
  \beditor{\binits{Isabelle} \bsnm{Guyon}}, and \beditor{\binits{Roman}
  \bsnm{Garnett}},
\bfpage{2235}--\blpage{2243}.
\burl{http://papers.nips.cc/paper/6280-deepmath-deep-sequence-models-for-premise-selection}.
\end{bchapter}
\endbibitem

\bibitem[\protect\citeauthoryear{Jakubuv and Urban}{2017}]{jakubuv2017-enigma}
\begin{bchapter}
\bauthor{\bsnm{Jakubuv}, \binits{Jan}}, and \bauthor{\binits{Josef}
  \bsnm{Urban}}.
\byear{2017}.
\bctitle{{ENIGMA:} efficient learning-based inference guiding machine}.
In \bbtitle{{CICM}},
eds. \beditor{\binits{Herman} \bsnm{Geuvers}}, \beditor{\binits{Matthew}
  \bsnm{England}}, \beditor{\binits{Osman} \bsnm{Hasan}},
  \beditor{\binits{Florian} \bsnm{Rabe}}, and \beditor{\binits{Olaf}
  \bsnm{Teschke}}.
Vol. \bseriesno{10383}
of \bsertitle{{LNCS}},
\bfpage{292}--\blpage{302}.
\bpublisher{Springer}.
doi:\doiurl{10.1007/978-3-319-62075-6_20}.
\bisbn{978-3-319-62074-9}.
\end{bchapter}
\endbibitem

\bibitem[\protect\citeauthoryear{Jones}{1973}]{jones1973-idf}
\begin{barticle}
\bauthor{\bsnm{Jones}, \binits{Karen~Spärck}}.
\byear{1973}.
\batitle{Index term weighting}.
\bjtitle{Information Storage and Retrieval}
\bvolume{9} (\bissue{11}): \bfpage{619}--\blpage{633}.
doi:\doiurl{10.1016/0020-0271(73)90043-0}.
\end{barticle}
\endbibitem

\bibitem[\protect\citeauthoryear{Kaliszyk and
  Urban}{2014}]{kaliszyk2014-holyhammer}
\begin{barticle}
\bauthor{\bsnm{Kaliszyk}, \binits{Cezary}}, and \bauthor{\binits{Josef}
  \bsnm{Urban}}.
\byear{2014}.
\batitle{Learning-assisted automated reasoning with {Flyspeck}}.
\bjtitle{J. Autom. Reasoning}
\bvolume{53} (\bissue{2}): \bfpage{173}--\blpage{213}.
doi:\doiurl{10.1007/s10817-014-9303-3}.
\end{barticle}
\endbibitem

\bibitem[\protect\citeauthoryear{Kaliszyk and
  Urban}{2015a}]{kaliszyk2015-femalecop}
\begin{bchapter}
\bauthor{\bsnm{Kaliszyk}, \binits{Cezary}}, and \bauthor{\binits{Josef}
  \bsnm{Urban}}.
\byear{2015}a.
\bctitle{{FEMaLeCoP}: Fairly efficient machine learning connection prover}.
In \bbtitle{{LPAR-20}},
eds. \beditor{\binits{Martin} \bsnm{Davis}}, \beditor{\binits{Ansgar}
  \bsnm{Fehnker}}, \beditor{\binits{Annabelle} \bsnm{McIver}}, and
  \beditor{\binits{Andrei} \bsnm{Voronkov}}.
Vol. \bseriesno{9450}
of \bsertitle{{LNCS}},
\bfpage{88}--\blpage{96}.
\bpublisher{Springer}.
doi:\doiurl{10.1007/978-3-662-48899-7_7}.
\bisbn{978-3-662-48898-0}.
\end{bchapter}
\endbibitem

\bibitem[\protect\citeauthoryear{Kaliszyk and
  Urban}{2015b}]{kaliszyk2015-mizar40}
\begin{barticle}
\bauthor{\bsnm{Kaliszyk}, \binits{Cezary}}, and \bauthor{\binits{Josef}
  \bsnm{Urban}}.
\byear{2015}b.
\batitle{{MizAR} 40 for {Mizar} 40}.
\bjtitle{J. Autom. Reasoning}
\bvolume{55} (\bissue{3}): \bfpage{245}--\blpage{256}.
doi:\doiurl{10.1007/s10817-015-9330-8}.
\end{barticle}
\endbibitem

\bibitem[\protect\citeauthoryear{Kaliszyk et~al.}{2015a}]{kaliszyk2015-holcop}
\begin{bchapter}
\bauthor{\bsnm{Kaliszyk}, \binits{Cezary}}, \bauthor{\binits{Josef}
  \bsnm{Urban}}, and \bauthor{\binits{Jir{\'{\i}}} \bsnm{Vyskocil}}.
\byear{2015}a.
\bctitle{Certified connection tableaux proofs for {HOL} {Light} and {TPTP}}.
In \bbtitle{{CPP}},
eds. \beditor{\binits{Xavier} \bsnm{Leroy}} and \beditor{\binits{Alwen}
  \bsnm{Tiu}},
\bfpage{59}--\blpage{66}.
\bpublisher{{ACM}}.
doi:\doiurl{10.1145/2676724.2693176}.
\bisbn{978-1-4503-3296-5}.
\end{bchapter}
\endbibitem

\bibitem[\protect\citeauthoryear{Kaliszyk
  et~al.}{2015b}]{kaliszyk2015-features}
\begin{bchapter}
\bauthor{\bsnm{Kaliszyk}, \binits{Cezary}}, \bauthor{\binits{Josef}
  \bsnm{Urban}}, and \bauthor{\binits{Jir{\'{\i}}} \bsnm{Vyskocil}}.
\byear{2015}b.
\bctitle{Efficient semantic features for automated reasoning over large
  theories}.
In \bbtitle{{IJCAI}},
eds. \beditor{\binits{Qiang} \bsnm{Yang}} and \beditor{\binits{Michael}
  \bsnm{Wooldridge}},
\bfpage{3084}--\blpage{3090}.
\bpublisher{{AAAI} Press}.
\bisbn{978-1-57735-738-4}.
\burl{http://ijcai.org/Abstract/15/435}.
\end{bchapter}
\endbibitem

\bibitem[\protect\citeauthoryear{Kaliszyk et~al.}{2015c}]{kaliszyk2015-et}
\begin{bchapter}
\bauthor{\bsnm{Kaliszyk}, \binits{Cezary}}, \bauthor{\binits{Stephan}
  \bsnm{Schulz}}, \bauthor{\binits{Josef} \bsnm{Urban}}, and
  \bauthor{\binits{Jir{\'{\i}}} \bsnm{Vyskocil}}.
\byear{2015}c.
\bctitle{System description: {E.T.} 0.1}.
In \bbtitle{{CADE-25}},
eds. \beditor{\binits{Amy~P.} \bsnm{Felty}} and \beditor{\binits{Aart}
  \bsnm{Middeldorp}}.
Vol. \bseriesno{9195}
of \bsertitle{{LNCS}},
\bfpage{389}--\blpage{398}.
\bpublisher{Springer}.
doi:\doiurl{10.1007/978-3-319-21401-6_27}.
\bisbn{978-3-319-21400-9}.
\end{bchapter}
\endbibitem

\bibitem[\protect\citeauthoryear{Kocsis and
  Szepesv{\'{a}}ri}{2006}]{kocsis2006-uct}
\begin{bchapter}
\bauthor{\bsnm{Kocsis}, \binits{Levente}}, and \bauthor{\binits{Csaba}
  \bsnm{Szepesv{\'{a}}ri}}.
\byear{2006}.
\bctitle{Bandit based {Monte-Carlo} planning}.
In \bbtitle{{ECML}},
eds. \beditor{\binits{Johannes} \bsnm{F{\"{u}}rnkranz}},
  \beditor{\binits{Tobias} \bsnm{Scheffer}}, and \beditor{\binits{Myra}
  \bsnm{Spiliopoulou}}.
Vol. \bseriesno{4212}
of \bsertitle{{LNCS}},
\bfpage{282}--\blpage{293}.
\bpublisher{Springer}.
doi:\doiurl{10.1007/11871842_29}.
\bisbn{3-540-45375-X}.
\end{bchapter}
\endbibitem

\bibitem[\protect\citeauthoryear{Kov{\'{a}}cs and
  Voronkov}{2013}]{kovacs2013-vampire}
\begin{bchapter}
\bauthor{\bsnm{Kov{\'{a}}cs}, \binits{Laura}}, and \bauthor{\binits{Andrei}
  \bsnm{Voronkov}}.
\byear{2013}.
\bctitle{First-order theorem proving and {Vampire}}.
In \bbtitle{{CAV}},
eds. \beditor{\binits{Natasha} \bsnm{Sharygina}} and \beditor{\binits{Helmut}
  \bsnm{Veith}}.
Vol. \bseriesno{8044}
of \bsertitle{{LNCS}},
\bfpage{1}--\blpage{35}.
\bpublisher{Springer}.
doi:\doiurl{10.1007/978-3-642-39799-8_1}.
\bisbn{978-3-642-39798-1}.
\end{bchapter}
\endbibitem

\bibitem[\protect\citeauthoryear{K{\"{u}}hlwein
  et~al.}{2012}]{kuehlwein2012-premsel}
\begin{bchapter}
\bauthor{\bsnm{K{\"{u}}hlwein}, \binits{Daniel}}, \bauthor{\binits{Twan}
  \bparticle{van }\bsnm{Laarhoven}}, \bauthor{\binits{Evgeni}
  \bsnm{Tsivtsivadze}}, \bauthor{\binits{Josef} \bsnm{Urban}}, and
  \bauthor{\binits{Tom} \bsnm{Heskes}}.
\byear{2012}.
\bctitle{Overview and evaluation of premise selection techniques for large
  theory mathematics}.
In \bbtitle{{IJCAR}},
eds. \beditor{\binits{Bernhard} \bsnm{Gramlich}}, \beditor{\binits{Dale}
  \bsnm{Miller}}, and \beditor{\binits{Uli} \bsnm{Sattler}}.
Vol. \bseriesno{7364}
of \bsertitle{{LNCS}},
\bfpage{378}--\blpage{392}.
\bpublisher{Springer}.
doi:\doiurl{10.1007/978-3-642-31365-3_30}.
\bisbn{978-3-642-31364-6}.
\end{bchapter}
\endbibitem

\bibitem[\protect\citeauthoryear{K{\"{u}}hlwein
  et~al.}{2013}]{kuehlwein2013-mash}
\begin{bchapter}
\bauthor{\bsnm{K{\"{u}}hlwein}, \binits{Daniel}},
  \bauthor{\binits{Jasmin~Christian} \bsnm{Blanchette}},
  \bauthor{\binits{Cezary} \bsnm{Kaliszyk}}, and \bauthor{\binits{Josef}
  \bsnm{Urban}}.
\byear{2013}.
\bctitle{{MaSh}: Machine learning for {Sledgehammer}}.
In \bbtitle{{ITP}},
eds. \beditor{\binits{Sandrine} \bsnm{Blazy}}, \beditor{\binits{Christine}
  \bsnm{Paulin-Mohring}}, and \beditor{\binits{David} \bsnm{Pichardie}}.
Vol. \bseriesno{7998}
of \bsertitle{{LNCS}},
\bfpage{35}--\blpage{50}.
\bpublisher{Springer}.
doi:\doiurl{10.1007/978-3-642-39634-2_6}.
\end{bchapter}
\endbibitem

\bibitem[\protect\citeauthoryear{Letz and Stenz}{2001}]{letz2001-connection}
\begin{bchapter}
\bauthor{\bsnm{Letz}, \binits{Reinhold}}, and \bauthor{\binits{Gernot}
  \bsnm{Stenz}}.
\byear{2001}.
\bctitle{Model elimination and connection tableau procedures}.
In \bbtitle{Handbook of automated reasoning (in 2 volumes)},
eds. \beditor{\binits{John~Alan} \bsnm{Robinson}} and \beditor{\binits{Andrei}
  \bsnm{Voronkov}},
\bfpage{2015}--\blpage{2114}.
\bpublisher{Elsevier and {MIT} Press}.
\bisbn{0-444-50813-9}.
\end{bchapter}
\endbibitem

\bibitem[\protect\citeauthoryear{Loos et~al.}{2017}]{loos2017-deepnetwork}
\begin{bchapter}
\bauthor{\bsnm{Loos}, \binits{Sarah~M.}}, \bauthor{\binits{Geoffrey}
  \bsnm{Irving}}, \bauthor{\binits{Christian} \bsnm{Szegedy}}, and
  \bauthor{\binits{Cezary} \bsnm{Kaliszyk}}.
\byear{2017}.
\bctitle{Deep network guided proof search}.
In \bbtitle{{LPAR-21}},
eds. \beditor{\binits{Thomas} \bsnm{Eiter}} and \beditor{\binits{David}
  \bsnm{Sands}}.
Vol. \bseriesno{46}
of \bsertitle{Epic series in computing},
\bfpage{85}--\blpage{105}.
\bpublisher{EasyChair}.
\burl{http://www.easychair.org/publications/paper/340345}.
\end{bchapter}
\endbibitem

\bibitem[\protect\citeauthoryear{Meng and Paulson}{2009}]{meng2009-mepo}
\begin{barticle}
\bauthor{\bsnm{Meng}, \binits{Jia}}, and \bauthor{\binits{Lawrence~C.}
  \bsnm{Paulson}}.
\byear{2009}.
\batitle{Lightweight relevance filtering for machine-generated resolution
  problems}.
\bjtitle{J. Applied Logic}
\bvolume{7} (\bissue{1}): \bfpage{41}--\blpage{57}.
doi:\doiurl{10.1016/j.jal.2007.07.004}.
\end{barticle}
\endbibitem

\bibitem[\protect\citeauthoryear{Mohamed et~al.}{2008}]{DBLP:conf/tphol/2008}
\begin{bbook}
\beditor{\bsnm{Mohamed}, \binits{Otmane~A{\"{\i}}t}},
  \beditor{\binits{C{\'{e}}sar~A.} \bsnm{Mu{\~{n}}oz}}, and
  \beditor{\binits{Sofi{\`{e}}ne} \bsnm{Tahar}}, eds.
\byear{2008}.
\bctitle{{TPHOLs}}.
Vol. \bseriesno{5170}
of \bsertitle{{LNCS}}.
\bpublisher{Springer}.
\bisbn{978-3-540-71065-3}.
\end{bbook}
\endbibitem

\bibitem[\protect\citeauthoryear{Olivetti and
  Tiwari}{2016}]{DBLP:conf/cade/2016}
\begin{bbook}
\beditor{\bsnm{Olivetti}, \binits{Nicola}}, and \beditor{\binits{Ashish}
  \bsnm{Tiwari}}, eds.
\byear{2016}.
\bctitle{{IJCAR}}.
Vol. \bseriesno{9706}
of \bsertitle{{LNCS}}.
\bpublisher{Springer}.
doi:\doiurl{10.1007/978-3-319-40229-1}.
\bisbn{978-3-319-40228-4}.
\end{bbook}
\endbibitem

\bibitem[\protect\citeauthoryear{Otten}{2005}]{otten2005-ileancop}
\begin{bchapter}
\bauthor{\bsnm{Otten}, \binits{Jens}}.
\byear{2005}.
\bctitle{Clausal connection-based theorem proving in intuitionistic first-order
  logic}.
In \bbtitle{{TABLEAUX}},
ed. \beditor{\binits{Bernhard} \bsnm{Beckert}}.
Vol. \bseriesno{3702}
of \bsertitle{{LNCS}},
\bfpage{245}--\blpage{261}.
\bpublisher{Springer}.
doi:\doiurl{10.1007/11554554_19}.
\bisbn{3-540-28931-3}.
\end{bchapter}
\endbibitem

\bibitem[\protect\citeauthoryear{Otten}{2008}]{otten2008-leancop}
\begin{bchapter}
\bauthor{\bsnm{Otten}, \binits{Jens}}.
\byear{2008}.
\bctitle{{leanCoP} 2.0 and {ileanCoP} 1.2: High performance lean theorem
  proving in classical and intuitionistic logic (system descriptions)}.
In \bbtitle{{IJCAR}},
eds. \beditor{\binits{Alessandro} \bsnm{Armando}}, \beditor{\binits{Peter}
  \bsnm{Baumgartner}}, and \beditor{\binits{Gilles} \bsnm{Dowek}}.
Vol. \bseriesno{5195}
of \bsertitle{{LNCS}},
\bfpage{283}--\blpage{291}.
\bpublisher{Springer}.
doi:\doiurl{10.1007/978-3-540-71070-7_23}.
\bisbn{978-3-540-71069-1}.
\end{bchapter}
\endbibitem

\bibitem[\protect\citeauthoryear{Otten}{2010}]{otten2010-cut}
\begin{barticle}
\bauthor{\bsnm{Otten}, \binits{Jens}}.
\byear{2010}.
\batitle{Restricting backtracking in connection calculi}.
\bjtitle{{AI} Commun.}
\bvolume{23} (\bissue{2-3}): \bfpage{159}--\blpage{182}.
doi:\doiurl{10.3233/AIC-2010-0464}.
\end{barticle}
\endbibitem

\bibitem[\protect\citeauthoryear{Otten}{2011}]{otten2011-nonclausal}
\begin{bchapter}
\bauthor{\bsnm{Otten}, \binits{Jens}}.
\byear{2011}.
\bctitle{A non-clausal connection calculus}.
In \bbtitle{{TABLEAUX}},
eds. \beditor{\binits{Kai} \bsnm{Br{\"{u}}nnler}} and \beditor{\binits{George}
  \bsnm{Metcalfe}}.
Vol. \bseriesno{6793}
of \bsertitle{{LNCS}},
\bfpage{226}--\blpage{241}.
\bpublisher{Springer}.
doi:\doiurl{10.1007/978-3-642-22119-4_18}.
\bisbn{978-3-642-22118-7}.
\end{bchapter}
\endbibitem

\bibitem[\protect\citeauthoryear{Otten}{2014}]{otten2014-mleancop}
\begin{bchapter}
\bauthor{\bsnm{Otten}, \binits{Jens}}.
\byear{2014}.
\bctitle{Mleancop: {A} connection prover for first-order modal logic}.
In \bbtitle{{IJCAR}},
eds. \beditor{\binits{St{\'{e}}phane} \bsnm{Demri}}, \beditor{\binits{Deepak}
  \bsnm{Kapur}}, and \beditor{\binits{Christoph} \bsnm{Weidenbach}}.
Vol. \bseriesno{8562}
of \bsertitle{{LNCS}},
\bfpage{269}--\blpage{276}.
\bpublisher{Springer}.
doi:\doiurl{10.1007/978-3-319-08587-6_20}.
\bisbn{978-3-319-08586-9}.
\end{bchapter}
\endbibitem

\bibitem[\protect\citeauthoryear{Otten}{2016}]{otten2016-nanocop}
\begin{bchapter}
\bauthor{\bsnm{Otten}, \binits{Jens}}.
\byear{2016}.
\bctitle{{nanoCoP}: {A} non-clausal connection prover}.
In \bbtitle{{IJCAR}},
eds. \beditor{\binits{Nicola} \bsnm{Olivetti}} and \beditor{\binits{Ashish}
  \bsnm{Tiwari}}.
Vol. \bseriesno{9706}
of \bsertitle{{LNCS}},
\bfpage{302}--\blpage{312}.
\bpublisher{Springer}.
doi:\doiurl{10.1007/978-3-319-40229-1_21}.
\bisbn{978-3-319-40228-4}.
\end{bchapter}
\endbibitem

\bibitem[\protect\citeauthoryear{Otten and Bibel}{2003}]{otten2003-leancop}
\begin{barticle}
\bauthor{\bsnm{Otten}, \binits{Jens}}, and \bauthor{\binits{Wolfgang}
  \bsnm{Bibel}}.
\byear{2003}.
\batitle{{leanCoP}: lean connection-based theorem proving}.
\bjtitle{J. Symb. Comput.}
\bvolume{36} (\bissue{1-2}): \bfpage{139}--\blpage{161}.
doi:\doiurl{10.1016/S0747-7171(03)00037-3}.
\end{barticle}
\endbibitem

\bibitem[\protect\citeauthoryear{Plaisted and
  Greenbaum}{1986}]{plaisted1986-cnf}
\begin{barticle}
\bauthor{\bsnm{Plaisted}, \binits{David~A.}}, and \bauthor{\binits{Steven}
  \bsnm{Greenbaum}}.
\byear{1986}.
\batitle{A structure-preserving clause form translation}.
\bjtitle{J. Symb. Comput.}
\bvolume{2} (\bissue{3}): \bfpage{293}--\blpage{304}.
doi:\doiurl{10.1016/S0747-7171(86)80028-1}.
\end{barticle}
\endbibitem

\bibitem[\protect\citeauthoryear{Plotkin}{1975}]{plotkin1975-continuation}
\begin{barticle}
\bauthor{\bsnm{Plotkin}, \binits{Gordon~D.}}
\byear{1975}.
\batitle{Call-by-name, call-by-value and the lambda-calculus}.
\bjtitle{Theor. Comput. Sci.}
\bvolume{1} (\bissue{2}): \bfpage{125}--\blpage{159}.
doi:\doiurl{10.1016/0304-3975(75)90017-1}.
\end{barticle}
\endbibitem

\bibitem[\protect\citeauthoryear{Ramakrishnan
  et~al.}{2001}]{ramakrishnan2001-indexing}
\begin{bchapter}
\bauthor{\bsnm{Ramakrishnan}, \binits{I.~V.}}, \bauthor{\binits{R.~C.}
  \bsnm{Sekar}}, and \bauthor{\binits{Andrei} \bsnm{Voronkov}}.
\byear{2001}.
\bctitle{Term indexing}.
In \bbtitle{Handbook of automated reasoning (in 2 volumes)},
eds. \beditor{\binits{John~Alan} \bsnm{Robinson}} and \beditor{\binits{Andrei}
  \bsnm{Voronkov}},
\bfpage{1853}--\blpage{1964}.
\bpublisher{Elsevier and {MIT} Press}.
\bisbn{0-444-50813-9}.
\end{bchapter}
\endbibitem

\bibitem[\protect\citeauthoryear{Raths and Otten}{2008}]{raths2008-randocop}
\begin{bchapter}
\bauthor{\bsnm{Raths}, \binits{Thomas}}, and \bauthor{\binits{Jens}
  \bsnm{Otten}}.
\byear{2008}.
\bctitle{{randoCoP}: Randomizing the proof search order in the connection
  calculus}.
In \bbtitle{{PAAR}},
eds. \beditor{\binits{Boris} \bsnm{Konev}}, \beditor{\binits{Renate~A.}
  \bsnm{Schmidt}}, and \beditor{\binits{Stephan} \bsnm{Schulz}}.
Vol. \bseriesno{373}
of \bsertitle{{CEUR} workshop proceedings}.
\bpublisher{CEUR-WS.org}.
\burl{http://ceur-ws.org/Vol-373/paper-08.pdf}.
\end{bchapter}
\endbibitem

\bibitem[\protect\citeauthoryear{Raths et~al.}{2007}]{raths2007-iltp}
\begin{barticle}
\bauthor{\bsnm{Raths}, \binits{Thomas}}, \bauthor{\binits{Jens} \bsnm{Otten}},
  and \bauthor{\binits{Christoph} \bsnm{Kreitz}}.
\byear{2007}.
\batitle{The {ILTP} problem library for intuitionistic logic}.
\bjtitle{J. Autom. Reasoning}
\bvolume{38} (\bissue{1-3}): \bfpage{261}--\blpage{271}.
doi:\doiurl{10.1007/s10817-006-9060-z}.
\end{barticle}
\endbibitem

\bibitem[\protect\citeauthoryear{Reis}{2015}]{reis2015-expansion}
\begin{bchapter}
\bauthor{\bsnm{Reis}, \binits{Giselle}}.
\byear{2015}.
\bctitle{Importing {SMT} and connection proofs as expansion trees}.
In \bbtitle{{PxTP}},
eds. \beditor{\binits{Cezary} \bsnm{Kaliszyk}} and \beditor{\binits{Andrei}
  \bsnm{Paskevich}}.
Vol. \bseriesno{186}
of \bsertitle{{EPTCS}},
\bfpage{3}--\blpage{10}.
doi:\doiurl{10.4204/EPTCS.186.3}.
\end{bchapter}
\endbibitem

\bibitem[\protect\citeauthoryear{Ridge and Margetson}{2005}]{ridge2005-verifol}
\begin{bchapter}
\bauthor{\bsnm{Ridge}, \binits{Tom}}, and \bauthor{\binits{James}
  \bsnm{Margetson}}.
\byear{2005}.
\bctitle{A mechanically verified, sound and complete theorem prover for first
  order logic}.
In \bbtitle{{TPHOLs}},
eds. \beditor{\binits{Joe} \bsnm{Hurd}} and \beditor{\binits{Thomas~F.}
  \bsnm{Melham}}.
Vol. \bseriesno{3603}
of \bsertitle{{LNCS}},
\bfpage{294}--\blpage{309}.
\bpublisher{Springer}.
doi:\doiurl{10.1007/11541868_19}.
\bisbn{3-540-28372-2}.
\end{bchapter}
\endbibitem

\bibitem[\protect\citeauthoryear{Robinson and
  Voronkov}{2001}]{DBLP:books/el/RobinsonV01}
\begin{bbook}
\beditor{\bsnm{Robinson}, \binits{John~Alan}}, and \beditor{\binits{Andrei}
  \bsnm{Voronkov}}, eds.
\byear{2001}.
\bbtitle{Handbook of automated reasoning (in 2 volumes)}.
\bpublisher{Elsevier and {MIT} Press}.
\end{bbook}
\endbibitem

\bibitem[\protect\citeauthoryear{Schadd et~al.}{2012}]{schadd2012-samegame}
\begin{barticle}
\bauthor{\bsnm{Schadd}, \binits{Maarten P.~D.}}, \bauthor{\binits{Mark H.~M.}
  \bsnm{Winands}}, \bauthor{\binits{Mandy J.~W.} \bsnm{Tak}}, and
  \bauthor{\binits{Jos W. H.~M.} \bsnm{Uiterwijk}}.
\byear{2012}.
\batitle{Single-player {Monte-Carlo} tree search for {SameGame}}.
\bjtitle{Knowl.-Based Syst.}
\bvolume{34}: \bfpage{3}--\blpage{11}.
doi:\doiurl{10.1016/j.knosys.2011.08.008}.
\end{barticle}
\endbibitem

\bibitem[\protect\citeauthoryear{Schmitt et~al.}{2001}]{schmitt2001-jprover}
\begin{bchapter}
\bauthor{\bsnm{Schmitt}, \binits{Stephan}}, \bauthor{\binits{Lori}
  \bsnm{Lorigo}}, \bauthor{\binits{Christoph} \bsnm{Kreitz}}, and
  \bauthor{\binits{Aleksey} \bsnm{Nogin}}.
\byear{2001}.
\bctitle{{JProver}: Integrating connection-based theorem proving into
  interactive proof assistants}.
In \bbtitle{{IJCAR}},
eds. \beditor{\binits{Rajeev} \bsnm{Gor{\'{e}}}}, \beditor{\binits{Alexander}
  \bsnm{Leitsch}}, and \beditor{\binits{Tobias} \bsnm{Nipkow}}.
Vol. \bseriesno{2083}
of \bsertitle{{LNCS}},
\bfpage{421}--\blpage{426}.
\bpublisher{Springer}.
doi:\doiurl{10.1007/3-540-45744-5_34}.
\bisbn{3-540-42254-4}.
\end{bchapter}
\endbibitem

\bibitem[\protect\citeauthoryear{Schulz}{2001}]{schulz2001-learning}
\begin{bchapter}
\bauthor{\bsnm{Schulz}, \binits{Stephan}}.
\byear{2001}.
\bctitle{Learning search control knowledge for equational theorem proving}.
In \bbtitle{{KI}},
eds. \beditor{\binits{Franz} \bsnm{Baader}}, \beditor{\binits{Gerhard}
  \bsnm{Brewka}}, and \beditor{\binits{Thomas} \bsnm{Eiter}}.
Vol. \bseriesno{2174}
of \bsertitle{{LNCS}},
\bfpage{320}--\blpage{334}.
\bpublisher{Springer}.
doi:\doiurl{10.1007/3-540-45422-5_23}.
\bisbn{3-540-42612-4}.
\end{bchapter}
\endbibitem

\bibitem[\protect\citeauthoryear{Schulz}{2013}]{schulz2013-e}
\begin{bchapter}
\bauthor{\bsnm{Schulz}, \binits{Stephan}}.
\byear{2013}.
\bctitle{System description: {E} 1.8}.
In \bbtitle{{LPAR-19}},
eds. \beditor{\binits{Kenneth~L.} \bsnm{McMillan}}, \beditor{\binits{Aart}
  \bsnm{Middeldorp}}, and \beditor{\binits{Andrei} \bsnm{Voronkov}}.
Vol. \bseriesno{8312}
of \bsertitle{{LNCS}},
\bfpage{735}--\blpage{743}.
\bpublisher{Springer}.
doi:\doiurl{10.1007/978-3-642-45221-5_49}.
\bisbn{978-3-642-45220-8}.
\end{bchapter}
\endbibitem

\bibitem[\protect\citeauthoryear{Silver et~al.}{2016}]{silver2016-alphago}
\begin{barticle}
\bauthor{\bsnm{Silver}, \binits{David}}, \bauthor{\binits{Aja} \bsnm{Huang}},
  \bauthor{\binits{Christopher~J.} \bsnm{Maddison}}, \bauthor{\binits{Arthur}
  \bsnm{Guez}}, \bauthor{\binits{Laurent} \bsnm{Sifre}},
  \bauthor{\binits{George} \bparticle{van~den }\bsnm{Driessche}},
  \bauthor{\binits{Julian} \bsnm{Schrittwieser}}, \bauthor{\binits{Ioannis}
  \bsnm{Antonoglou}}, \bauthor{\binits{Veda} \bsnm{Panneershelvam}},
  \bauthor{\binits{Marc} \bsnm{Lanctot}}, \bauthor{\binits{Sander}
  \bsnm{Dieleman}}, \bauthor{\binits{Dominik} \bsnm{Grewe}},
  \bauthor{\binits{John} \bsnm{Nham}}, \bauthor{\binits{Nal}
  \bsnm{Kalchbrenner}}, \bauthor{\binits{Ilya} \bsnm{Sutskever}},
  \bauthor{\binits{Timothy} \bsnm{Lillicrap}}, \bauthor{\binits{Madeleine}
  \bsnm{Leach}}, \bauthor{\binits{Koray} \bsnm{Kavukcuoglu}},
  \bauthor{\binits{Thore} \bsnm{Graepel}}, and \bauthor{\binits{Demis}
  \bsnm{Hassabis}}.
\byear{2016}.
\batitle{Mastering the game of {Go} with deep neural networks and tree search}.
\bjtitle{Nature}
\bvolume{529}: \bfpage{484}--\blpage{503}.
\burl{http://www.nature.com/nature/journal/v529/n7587/full/nature16961.html}.
\end{barticle}
\endbibitem

\bibitem[\protect\citeauthoryear{Slind and Norrish}{2008}]{slind2008-hol4}
\begin{bchapter}
\bauthor{\bsnm{Slind}, \binits{Konrad}}, and \bauthor{\binits{Michael}
  \bsnm{Norrish}}.
\byear{2008}.
\bctitle{A brief overview of {HOL4}}.
In \bbtitle{{TPHOLs}},
eds. \beditor{\binits{Otmane~A{\"{\i}}t} \bsnm{Mohamed}},
  \beditor{\binits{C{\'{e}}sar~A.} \bsnm{Mu{\~{n}}oz}}, and
  \beditor{\binits{Sofi{\`{e}}ne} \bsnm{Tahar}}.
Vol. \bseriesno{5170}
of \bsertitle{{LNCS}},
\bfpage{28}--\blpage{32}.
\bpublisher{Springer}.
doi:\doiurl{10.1007/978-3-540-71067-7_6}.
\bisbn{978-3-540-71065-3}.
\end{bchapter}
\endbibitem

\bibitem[\protect\citeauthoryear{Sutcliffe}{2009a}]{sutcliffe2009-j4}
\begin{barticle}
\bauthor{\bsnm{Sutcliffe}, \binits{Geoff}}.
\byear{2009}a.
\batitle{The 4th {IJCAR} automated theorem proving system competition -
  {CASC-J4}}.
\bjtitle{{AI} Commun.}
\bvolume{22} (\bissue{1}): \bfpage{59}--\blpage{72}.
doi:\doiurl{10.3233/AIC-2009-0441}.
\end{barticle}
\endbibitem

\bibitem[\protect\citeauthoryear{Sutcliffe}{2009b}]{sutcliffe2009-tptp}
\begin{barticle}
\bauthor{\bsnm{Sutcliffe}, \binits{Geoff}}.
\byear{2009}b.
\batitle{The {TPTP} problem library and associated infrastructure}.
\bjtitle{J. Autom. Reasoning}
\bvolume{43} (\bissue{4}): \bfpage{337}--\blpage{362}.
doi:\doiurl{10.1007/s10817-009-9143-8}.
\end{barticle}
\endbibitem

\bibitem[\protect\citeauthoryear{Sutcliffe}{2011}]{sutcliffe2011-j5}
\begin{barticle}
\bauthor{\bsnm{Sutcliffe}, \binits{Geoff}}.
\byear{2011}.
\batitle{The 5th {IJCAR} automated theorem proving system competition -
  {CASC-J5}}.
\bjtitle{{AI} Commun.}
\bvolume{24} (\bissue{1}): \bfpage{75}--\blpage{89}.
doi:\doiurl{10.3233/AIC-2010-0483}.
\end{barticle}
\endbibitem

\bibitem[\protect\citeauthoryear{Sutcliffe}{2016}]{sutcliffe2016-j8}
\begin{barticle}
\bauthor{\bsnm{Sutcliffe}, \binits{Geoff}}.
\byear{2016}.
\batitle{The 8th {IJCAR} automated theorem proving system competition -
  {CASC-J8}}.
\bjtitle{{AI} Commun.}
\bvolume{29} (\bissue{5}): \bfpage{607}--\blpage{619}.
doi:\doiurl{10.3233/AIC-160709}.
\end{barticle}
\endbibitem

\bibitem[\protect\citeauthoryear{Urban}{2004}]{urban2004-mptp}
\begin{barticle}
\bauthor{\bsnm{Urban}, \binits{Josef}}.
\byear{2004}.
\batitle{{MPTP} - motivation, implementation, first experiments}.
\bjtitle{J. Autom. Reasoning}
\bvolume{33} (\bissue{3-4}): \bfpage{319}--\blpage{339}.
doi:\doiurl{10.1007/s10817-004-6245-1}.
\end{barticle}
\endbibitem

\bibitem[\protect\citeauthoryear{Urban et~al.}{2011}]{urban2011-malecop}
\begin{bchapter}
\bauthor{\bsnm{Urban}, \binits{Josef}}, \bauthor{\binits{Jir{\'{\i}}}
  \bsnm{Vyskocil}}, and \bauthor{\binits{Petr} \bsnm{Step{\'{a}}nek}}.
\byear{2011}.
\bctitle{{MaLeCoP} machine learning connection prover}.
In \bbtitle{{TABLEAUX}},
eds. \beditor{\binits{Kai} \bsnm{Br{\"{u}}nnler}} and \beditor{\binits{George}
  \bsnm{Metcalfe}}.
Vol. \bseriesno{6793}
of \bsertitle{{LNCS}},
\bfpage{263}--\blpage{277}.
\bpublisher{Springer}.
doi:\doiurl{10.1007/978-3-642-22119-4_21}.
\bisbn{978-3-642-22118-7}.
\end{bchapter}
\endbibitem

\bibitem[\protect\citeauthoryear{Urban et~al.}{2008}]{urban2008-malarea}
\begin{bchapter}
\bauthor{\bsnm{Urban}, \binits{Josef}}, \bauthor{\binits{Geoff}
  \bsnm{Sutcliffe}}, \bauthor{\binits{Petr} \bsnm{Pudl{\'{a}}k}}, and
  \bauthor{\binits{Jir{\'{\i}}} \bsnm{Vyskocil}}.
\byear{2008}.
\bctitle{{MaLARea} {SG1-} machine learner for automated reasoning with semantic
  guidance}.
In \bbtitle{{IJCAR}},
eds. \beditor{\binits{Alessandro} \bsnm{Armando}}, \beditor{\binits{Peter}
  \bsnm{Baumgartner}}, and \beditor{\binits{Gilles} \bsnm{Dowek}}.
Vol. \bseriesno{5195}
of \bsertitle{{LNCS}},
\bfpage{441}--\blpage{456}.
\bpublisher{Springer}.
doi:\doiurl{10.1007/978-3-540-71070-7_37}.
\bisbn{978-3-540-71069-1}.
\end{bchapter}
\endbibitem

\bibitem[\protect\citeauthoryear{Veroff}{1996}]{veroff1996-hints}
\begin{barticle}
\bauthor{\bsnm{Veroff}, \binits{Robert}}.
\byear{1996}.
\batitle{Using hints to increase the effectiveness of an automated reasoning
  program: Case studies}.
\bjtitle{J. Autom. Reasoning}
\bvolume{16} (\bissue{3}): \bfpage{223}--\blpage{239}.
doi:\doiurl{10.1007/BF00252178}.
\end{barticle}
\endbibitem

\bibitem[\protect\citeauthoryear{Wenzel et~al.}{2008}]{wenzel2008-isabelle}
\begin{bchapter}
\bauthor{\bsnm{Wenzel}, \binits{Makarius}}, \bauthor{\binits{Lawrence~C.}
  \bsnm{Paulson}}, and \bauthor{\binits{Tobias} \bsnm{Nipkow}}.
\byear{2008}.
\bctitle{The {Isabelle} framework}.
In \bbtitle{{TPHOLs}},
eds. \beditor{\binits{Otmane~A{\"{\i}}t} \bsnm{Mohamed}},
  \beditor{\binits{C{\'{e}}sar~A.} \bsnm{Mu{\~{n}}oz}}, and
  \beditor{\binits{Sofi{\`{e}}ne} \bsnm{Tahar}}.
Vol. \bseriesno{5170}
of \bsertitle{{LNCS}},
\bfpage{33}--\blpage{38}.
\bpublisher{Springer}.
doi:\doiurl{10.1007/978-3-540-71067-7_7}.
\bisbn{978-3-540-71065-3}.
\end{bchapter}
\endbibitem

\bibitem[\protect\citeauthoryear{Whalen}{2016}]{whalen2016-holophrasm}
\begin{botherref}
\oauthor{\bsnm{Whalen}, \binits{Daniel}}.
2016.
Holophrasm: a neural automated theorem prover for higher-order logic.
\textit{CoRR}
abs/1608.02644.
\url{http://arxiv.org/abs/1608.02644}.
\end{botherref}
\endbibitem

\end{thebibliography}

\end{document}